\documentclass[11pt,letterpaper]{article}
\usepackage[T1]{fontenc}
\usepackage[utf8]{inputenc}
\usepackage{lmodern}
\usepackage[margin=1in,headheight=14pt,headsep=17pt,footskip=28pt]{geometry}
\usepackage{amsmath,amssymb,amsthm,mathtools}
\usepackage{xcolor}
\usepackage{microtype}
\usepackage{array,booktabs,longtable}
\usepackage{tikz-cd}
\usepackage{enumitem}
\usepackage{needspace}
\usepackage[titles]{tocloft}
\usepackage{fancyhdr}
\usepackage[hidelinks,bookmarksnumbered]{hyperref}
\hypersetup{
 pdftitle={Bohnenblust-Hille Bound on the Boolean Cube: Exponent 2.943},
 pdfkeywords={Bohnenblust-Hille, Boolean cube, weighted bootstrap, hypercontractivity}
}
\setlist[enumerate]{leftmargin=*,itemsep=3pt,topsep=4pt}
\newtheorem{theorem}{Theorem}[section]
\newtheorem{lemma}[theorem]{Lemma}
\newtheorem{proposition}[theorem]{Proposition}
\newtheorem{corollary}[theorem]{Corollary}
\theoremstyle{remark}
\newtheorem{remark}[theorem]{Remark}
\newcommand{\E}{\mathbb E}
\newcommand{\Prob}{\mathbb P}
\newcommand{\C}{\mathbb C}

\newcommand{\ee}{\mathrm e}
\newcommand{\wh}{\widehat}
\newcommand{\Om}{\Omega}
\newcommand{\norm}[1]{\left\lVert #1\right\rVert}
\newcommand{\abs}[1]{\left\lvert #1\right\rvert}
\newcommand{\wtX}{\widetilde X}
\newcommand{\wtY}{\widetilde Y}
\newcommand{\Glow}{c_{\mathrm{low}}}
\newcommand{\Ghigh}{c_{\mathrm{high}}}
\newcommand{\Dlow}{D_r^{\mathrm{low}}}
\newcommand{\Dhigh}{D_r^{\mathrm{high}}}
\newcommand{\sect}[1]{\S\,\ref{#1}}
\newcolumntype{P}[1]{>{\raggedright\arraybackslash}p{#1}}

\tikzset{
  bh vertex/.style={
    circle, fill=black, draw=none,
    inner sep=0pt, outer sep=0pt, minimum size=3.2pt
  },
  bh label/.style={
    rectangle, fill=none, draw=none,
    inner sep=1.4pt, outer sep=0pt, font=\Large
  },
  bh above/.style={
    bh vertex, label={[bh label, label distance=2pt]above:{#1}}
  },
  bh below/.style={
    bh vertex, label={[bh label, label distance=2pt]below:{#1}}
  }
}

\tikzcdset{
  bh tree/.style={
    arrows={dash, line width=0.45pt},
    cells={nodes={rectangle, inner sep=0pt, outer sep=0pt,
                  anchor=center}},
    column sep={4.4mm,between origins},
    row sep={7mm,between origins}
  }
}

\newcommand{\ExampleRoot}{13}

\DeclareMathOperator{\Var}{Var}
\DeclareMathOperator{\Inf}{Inf}

\theoremstyle{definition}
\newtheorem{definition}[theorem]{Definition}
\newtheorem{conjecture}[theorem]{Conjecture}

\begin{document}
\thispagestyle{plain}
\begin{center}
 {\LARGE\bfseries An $m^{2.943}$ Bohnenblust--Hille Bound\\[4pt]
 on the Boolean Cube\par}
 \vspace{8pt}
 {\large Joseph Slote, Chun-Kai Tseng, Alexander Volberg\par}
 Revised September 16, 2026
\end{center}

\begin{abstract}
Let $q_m=2m/(m+1)$ and put
\[
 \beta_0=\frac{3}{2}+\frac{1}{\log 2}=2.9426950408\ldots,
\]
where $\log$ is the natural logarithm. We give a proof scheme showing that,
for every $\varepsilon>0$, there is $C_\varepsilon<\infty$ such that every
complex-valued function $f:\{-1,1\}^n\to\C$ of Fourier degree at most $m$
satisfies
\[
 \left(\sum_{A\subseteq[n]}\abs{\wh f(A)}^{q_m}\right)^{1/q_m}
 \le C_\varepsilon m^{\beta_0+\varepsilon}\norm{f}_\infty.
\]
The improvement over the
$m^9$ estimate of the earlier draft has two ingredients. The first three
Fourier levels are estimated at the scales $m$, $m^{3/2}$, and $m^{5/3}$.
These losses are encoded in the weight
$M^{\mu_r/r}r^B\min\{r,L_M\}^8$, $L_M\asymp_B \log (M+1)$, with $\mu_r=\lceil3r/2\rceil$ for $r\ge2$.
A parity-compatible central window handles $r<L_M$, while a shrinking
balanced window handles $r\ge L_M\asymp_B\log(M+1)$, no parity is used in this regime. The leading
high-degree loss--gain factor is $\ee\,2^{-B}$; careful uniform bounds
close the bootstrap for every $B>1/\log2$. The formula for $\mu_r$ is explained below.
\end{abstract}
\vspace{-5pt}
{\small\setlength{\parskip}{0pt}\tableofcontents}

\clearpage
\section*{Notation guide}
\phantomsection\addcontentsline{toc}{section}{Notation guide}
\label{sec:notation}
All $L^p$ norms and expectations on a cube use normalized counting measure;
coefficient $\ell^p$ norms use unnormalized sums. We distinguished  $\ee=2.7818\dots$
 and  italic $e$, which is a degree of a block. The fixed cutoff
constant $C_B$ is never used as a generic constant: $C(B)$ denotes a
constant depending only on $B$ that may change between occurrences.

\begingroup
\small
\setlength{\tabcolsep}{5pt}
\renewcommand{\arraystretch}{1.13}
\begin{longtable}{@{}P{0.22\textwidth}P{0.615\textwidth}P{0.12\textwidth}@{}}
\toprule
\textbf{Symbol} & \textbf{Meaning and role} & \textbf{Location} \\
\midrule
\endfirsthead
\multicolumn{3}{@{}l}{\textit{Notation guide (continued)}}\\[3pt]
\toprule
\textbf{Symbol} & \textbf{Meaning and role} & \textbf{Location} \\
\midrule
\endhead
\midrule
\multicolumn{3}{r@{}}{\textit{Continued on the next page}}\\
\endfoot
\bottomrule
\endlastfoot
\multicolumn{3}{@{}l}{\textbf{Fourier data and degree parameters}}\\*[2pt]
$\Om_n$, $w_A$ & $\{-1,1\}^n$ and its Walsh character $w_A(x)=\prod_{j\in A}x_j$. & \sect{sec:statement}\\
$\wh f(A)$, $f_r$ & Walsh coefficient and degree-$r$ homogeneous part of $f$. & \sect{sec:statement}\\
$m$, $M$ & $m$: terminal degree in the BH statement; $M$: fixed degree cap during the bootstrap. & \sect{sec:statement}\\
$n$, $N$ & Actual dimension and a temporary finite upper bound on dimension. & \sect{sec:weighted}\\
$q_r$, $A_r(F)$ & $q_r=2r/(r+1)$; $A_r(F)$ is the coefficient $\ell^{q_r}$ norm on level $r$, not the full coefficient norm. & \eqref{eq:Ar}\\
$\beta_0$ & $\frac32+\frac{1}{\log2}$, the limiting exponent in the stated bound. & \eqref{eq:main}\\
$T_\rho$ & Noise operator: it multiplies level $r$ by $\rho^r$. & \sect{sec:noise}\\
$C_1,C_2,C_3$ & Absolute constants in the first-, second-, and third-level estimates. & \eqref{eq:first}, \eqref{eq:second}, \eqref{eq:third}\\
\addlinespace[5pt]
\multicolumn{3}{@{}l}{\textbf{Two-block coefficient arrays}}\\*[2pt]
$d,e,r,\theta$ & Block degrees $d,e$, total degree $r=d+e$, and interpolation weight $\theta=d/r$. & \sect{sec:mixed}\\
$Q_{d,e}$, $a(Q_{d,e})$ & Part of $Q(x,y)$ of degree $d$ in $x$ and $e$ in $y$, and its coefficient vector. & \sect{sec:weighted}\\
$X_{d,e}$, $Y_{d,e}$ & Row norm $\ell_A^{q_d}(\ell_B^2)$ and column norm $\ell_B^{q_e}(\ell_A^2)$. & Lemma~\ref{lem:two-block}\\
$\rho_s$ & $\sqrt{(s-1)/(s+1)}=\sqrt{q_s-1}$, defined for $s\ge2$. & \sect{sec:weighted}\\
$\wtX_{d,e}$, $\wtY_{d,e}$ & $\rho_d^e X_{d,e}$ for $d\ge2,e\ge0$; $\rho_e^dY_{d,e}$ for $d\ge0,e\ge2$. & Lemma~\ref{lem:simultaneous}\\
$B_{d,A}(y)$, $g_d(y)$ & The $x$-coefficient function with index $A$; $g_d(y)=A_d(Q(\cdot,y))$. & Lemma~\ref{lem:simultaneous}\\
$R(d,e)$ & $\rho_d^{-ed/r}\rho_e^{-de/r}$, the exact cost of removing the two noise factors. & \eqref{eq:deweight}\\
\addlinespace[5pt]
\multicolumn{3}{@{}l}{\textbf{Weights and fixed bootstrap parameters}}\\*[2pt]
$B$, $\eta_B$ & Weight exponent $B>1/\log2$ and positive margin $\eta_B=B\log2-1$. & \eqref{eq:eta}\\
$C_B$, $R_B$ & Fixed constants depending only on $B$: $C_B\ge8/\eta_B$ and a sufficiently large integer $R_B\ge6^4$. & \eqref{eq:cutoff}, \eqref{eq:RB}\\
$L_M$ & $\lceil\max\{R_B,C_B\log(M+1)\}\rceil$, the transition and saturation scale. & \eqref{eq:cutoff}\\
$\mu_r$, $\mu_r/r$ & Degree potential and actual exponent of $M$ in the weight. Here $\mu_1=1$ and $\mu_r=\lceil3r/2\rceil$ for $r\ge2$. & \eqref{eq:mu}\\
$\sigma_{M,r}$ & $\min\{r,L_M\}^8$, the auxiliary factor that saturates logarithmically. & \eqref{eq:sigma}\\
$w_{M,r}$ & $M^{\mu_r/r}r^B\sigma_{M,r}$, the full level weight. & \eqref{eq:weight}\\
$W_M(F)$ & $\bigl(\sum_{r=1}^M A_r(F)^2/w_{M,r}^2\bigr)^{1/2}$. & \eqref{eq:WM}\\
$\Gamma_{M,N}(B)$ & Best ratio $W_M(F)/\norm{F}_\infty$ for degree $\le M$ and dimension $\le N$. Finite before uniform closure. & \eqref{eq:Gamma}\\
$K_B$ & Uniform bound for the weighted norm after closure; independent of $M,n,N$. & Theorem~\ref{thm:weighted}\\
\addlinespace[5pt]
\multicolumn{3}{@{}l}{\textbf{Random splitting and contraction}}\\*[2pt]
$D_r$, $p_r$, $k_r$ & Allowed values of $d$ (for fixed $M$); their binomial probability; their number. & \eqref{eq:capture-def}\\
$\omega$, $Q_\omega$ & A random bicoloring and the resulting two-block version of $F$, with unchanged supremum norm. & \sect{sec:abstract}\\
$\Dlow$, $\Dhigh$ & Parity-compatible central window for $r<L_M$; near-balanced window for $r\ge L_M$. & \eqref{eq:Dlow}, \eqref{eq:Dhigh}\\
$h(t)$ & Binary entropy $-t\log t-(1-t)\log(1-t)$. & \eqref{eq:entropy}\\
$K_M(r,d)$ & Product of probability, split-count, deweighting, and weight-ratio factors. & \eqref{eq:KM}\\
$\kappa_M(r,d)$ & $\sqrt{2\max\{d/r,1-d/r\}}\,K_M(r,d)$, the full local contraction multiplier. & \sect{sec:choices}\\
$c=c(M)$ & Supremum of $\kappa_M(r,d)$ over all retained splits; $0$ for an empty range. & \eqref{eq:c-def}\\
$\Glow(M)$, $\Ghigh(M)$ & Suprema of the same multiplier over the low and high ranges. Their \emph{maximum}, not sum, is $c(M)$. & \sect{sec:choices}\\
$c_*$, $c_B$ & $c_*=45/(2\ee^4)<1$; $c_B=\max\{c_*,\exp(-5\eta_B/8)\}<1$, a uniform bound for $c(M)$. & \sect{sec:combine}\\
$E_B(r)$ & Sum of finite-$r$ logarithmic errors in the high-range multiplier. & \eqref{eq:EB}\\
$U_{d,e},V_{d,e},G$ & $\wtX_{d,e}/w_{M,d}$, $\wtY_{d,e}/w_{M,e}$, and $G=\Gamma_{M,N}(B)\norm{F}_\infty$. & Prop.~\ref{prop:abstract}\\
$u_{d,e},v_{d,e}$ & Temporary normalized squared energies $U_{d,e}^2/G^2$, $V_{d,e}^2/G^2$ in the contraction proof. & Prop.~\ref{prop:abstract}\\
\addlinespace[5pt]
\multicolumn{3}{@{}l}{\textbf{Terminal interpolation}}\\*[2pt]
$E_r$ & $\norm{f_r}_2$; unrelated to the high-range error $E_B(r)$. & \sect{sec:terminal}\\
$u_r$, $v_r$ & $A_r(f)/(K_Bw_{m,r}\norm f_\infty)$ and $E_r/\norm f_\infty$; each has squared sum at most $1$. & \eqref{eq:budgets}\\
$z_r$ & $u_r^{r/m}v_r^{1-r/m}$ for $r<m$, and $z_m=u_m$. & \sect{sec:terminal}\\
$C(B)$, $C_\varepsilon$ & Generic constants depending only on the indicated parameter. These are distinct from the fixed cutoff constant $C_B$. & Throughout\\
\end{longtable}
\endgroup

\clearpage

\section{Introduction}

We were made aware of the paper \cite{I} on September 7. It was published online on August 27. It has the same result with a slightly worse exponent.

\medskip

Bohnenblust--Hille inequality has an old and venerable history. It was first proved in \cite{BohnenblustHille} in 1931 and gave the 
$n$-free estimate of a certain coefficient norm of an analytic  polynomial  of variables $n$ via its maximum on multi-torus $\mathbb T^n$.
The constant was independent of $n$ but dependent in a superexponential fashion on the degree $m$ of the polynomial.
The result was needed to answer a question of Harold Bohr concerning a convergence property of Dirichlet series.

Much later this result was revised and the constant was improved in \cite{DFOOS}, and the constant became exponential in degreee $m$.

A bit later \cite{BayartPellegrinoSeoane-Sep\'ulveda} proved a subexponential in $m$ estimate for multi-torus case.

Very soon \cite{DefantMastyloPerez} proved a similar subexponential estimate for polynomials on Hamming cube.

This latter result was used in a learning theory by Eskenazis--Ivanisvili \cite{EskenazisIvanisvili} 
to obtain an optimal number of queries to PAC (probably approximately correct) learning of function of low degree on Boolean cube.

Below we give a polynomial in degree estimate for the constant of the Bohnenblast--Hille inequality on Boolean cube.

The argument below adapts the weighted graded square-function method introduced in
Pellegrino and Teixeira \cite{PellegrinoTeixeira}  Section 5. There it was introduced  for analytic polynomials.  
There are several new and exciting ideas in \cite{PellegrinoTeixeira}. One of them is a bicoloring of variables, another is a special square function
on coefficients of polynomials. 

On the Boolean cube,
Bonami--Beckner hypercontractivity supplies simultaneous row and column
control after a random bicoloring of the coordinates. Bicoloring idea comes from \cite{PellegrinoTeixeira}, but in the case of Hamming cube it can be realized in a simpler way.
Every monomial on Hamming cube is 
square-free, so the dominant-power compression regime from the analytic
setting disappears.  

In a previous version of this note the random split retains 
a fixed fraction of the
coefficient of monomial of degree $r$  in the central  window $d, e\in [r/4, 3r/4], d+e =r$, where $d, e$ are degrees of a monomial after bicoloring. This  givies a strict contraction that leads to a bootstrapping estimate for the Bohnenblust--Hille constant on Hamming cube.  Only the first five Fourier levels are left outside the
contraction; a Bernstein estimate for their spectral projections costs
$O(m^5)$.  The exponent $9$ was explicit and was deliberately unoptimized.

In this version a considerably more sophisticated square function is used and this allows for the improvement of the constant.

As a consequence, the Aaronson--Ambainis influence conjecture holds for
bounded polynomials whose Fourier coefficients are approximately flat within
each homogeneous level: if the within-level distortion is at most $\Lambda$,
then every coordinate has influence at least
$Var(f)^2/(C^2\Lambda^4 d^{s})$ with relatively small $s$.

In particular, if all Fourier--Walsh coefficients of a polynomial have the same absolute value (the Fourier flat case) then Aaronson--Ambainis influence conjecture holds.
Let us remind that the combination of \cite{ODonnellSaksSchrammServedio} and \cite{NisanSzegedy} 
proves that for real flat function $f$ on Hamming cube (say, Boolean $f$,  $|f|=1$) Aaronson--Ambainis influence conjecture also holds. See also \cite{Lee} for another  proof of this fact.

An appendix also records an exponential
counterexample to direct $L^\infty$ control of a single bidegree slice.

The proof is modeled on the weighted bootstrap in
\cite[Sections 3 and 5]{PellegrinoTeixeira}.  

Its Boolean specialization is based on a careful analysis of of levels $2$ and $3$ BH estimates plus the
variable splitting window changing completely its nature for $r\lesssim \log M$ and for larger values of $r$.
Accordingly the square function is more complicated than in  Pellegrino--Teixeira's~\cite{PellegrinoTeixeira}.



\section{Statement and proof design}
\label{sec:statement}
Write $\Om_n=\{-1,1\}^n$ with normalized counting measure. For
$A\subseteq[n]$, let
\[
 w_A(x)=\prod_{j\in A}x_j,
 \qquad \wh f(A)=\E_{x\in\Om_n}f(x)w_A(x).
\]
Then
\[
 f=\sum_{A\subseteq[n]}\wh f(A)w_A
   =\sum_{r=0}^n f_r,
 \qquad f_r=\sum_{|A|=r}\wh f(A)w_A.
\]
For $r\ge1$, put
\begin{equation}\label{eq:Ar}
 q_r=\frac{2r}{r+1},
 \qquad A_r(f)=\left(\sum_{|A|=r}\abs{\wh f(A)}^{q_r}\right)^{1/q_r}.
\end{equation}
Levels above the dimension or the degree are understood to be zero.
Throughout, $\log$ denotes the natural logarithm, and
$\beta_0=\frac32+\frac1{\log2}$.

\begin{theorem}[Boolean Bohnenblust--Hille bound]\label{thm:main}
For every $\varepsilon>0$ there is a constant $C_\varepsilon<\infty$ such
that, for all integers $m,n\ge1$ and every complex-valued $f:\Om_n\to\C$
of degree at most $m$,
\begin{equation}\label{eq:main}
 \left(\sum_{A\subseteq[n]}\abs{\wh f(A)}^{q_m}\right)^{1/q_m}
 \le C_\varepsilon m^{\beta_0+\varepsilon}\norm f_\infty.
\end{equation}
In particular, there is an absolute constant $C$ such that the left side is
at most $Cm^{2.943}\norm f_\infty$.
\end{theorem}

The proof retains the random-bicoloring and simultaneous row/column
estimates of the earlier Boolean-cube draft~\cite{earlier}, modeled on the
weighted analytic bootstrap of Pellegrino--Teixeira~\cite{PellegrinoTeixeira}. These
references provide context; the estimates from the earlier draft needed
below are proved here. For the established Boolean-cube setting, see
also~\cite{DefantMastyloPerez}.

The main change is the weight. For a polynomial of degree at most $M$, use
\begin{equation}\label{eq:weight-intro}
 w_{M,r}=M^{\mu_r/r}r^B\min\{r,L_M\}^8,
 \qquad \mu_1=1,\qquad
 \mu_r=\left\lceil\frac{3r}{2}\right\rceil\quad(r\ge2),
\end{equation}
where $B>1/\log2$ and $L_M\asymp_B\log(M+1)$. The associated square
function is
\begin{equation}\label{eq:WM}
 W_M(F)=\left(\sum_{r=1}^M\frac{A_r(F)^2}{w_{M,r}^2}\right)^{1/2}.
\end{equation}

\begin{remark}
The definition of $ \mathcal W_B(F)$ seems overcomplicated on the first glance. But in fact, its form has a simple explanation.
This grading in  homogeneity $r$ and keeping all $r\in [1, M]$ together is inevitable by the following reason.  
As we cannot claim that $Q_{\omega, d, e}$ obtained from bounded $F$ has a nice $L^\infty$  norm for fixed $d, e$ (see Addendum A for this claim), 
one is forced to consider the whole $Q_{\omega}$ as the source of nice norm, but $Q_\omega $
has of course many homogeneity modes; 
so one should invent a method to estimate all modes simultaneously—which explains the square function $W_B(F)$ in formula \eqref{eq:WM}
 as a vehicle for 
subsequent estimates. 
\end{remark}

The proof has five parts.
\begin{enumerate}
\item Fixed-level estimates give
$A_1(F)\lesssim M\norm F_\infty$,
$A_2(F)\lesssim M^{3/2}\norm F_\infty$, and
$A_3(F)\lesssim M^{5/3}\norm F_\infty$.
\item The seed exponents are encoded by $\mu_r$. For $d,e\ge2$, its
additivity defect $\mu_d+\mu_e-\mu_{d+e}$ is either $0$ or $1$.
\item For $4\le r<L_M$, a parity-compatible central split makes the defect
zero. The unsaturated auxiliary factor $r^8$ supplies extra entropy gain.
\item For $r\ge L_M$, all near-balanced splits are retained. The possible
defect costs at most $M^{1/r}\le\exp(1/C_B)$, with $C_B$ chosen large
depending only on $B$. The remaining finite-$r$ errors vanish as
$r\to\infty$, leaving the leading factor $\ee\,2^{-B}<1$.
\item A single uniform contraction is absorbed in $W_M$. Terminal
interpolation converts this weighted estimate into the BH inequality,
with a factor $m^{3/2+B}$ up to logarithms.
\end{enumerate}

\section{Noise and fixed-level estimates}\label{sec:noise}
For $-1\le\rho\le1$, define
\[
 T_\rho f=\sum_{A\subseteq[n]}\rho^{|A|}\wh f(A)w_A,
\]
with the constant coefficient unchanged, including when $\rho=0$.
For $0\le\rho\le1$, this is a positive averaging operator, hence
\begin{equation}\label{eq:noise-infty}
 \norm{T_\rho f}_\infty\le\norm f_\infty.
\end{equation}
We use the Bonami--Beckner inequality in the following form
\cite{Bonami,ODonnell}.
\begin{theorem}[Bonami--Beckner]\label{thm:BB}
Let $1\le p\le2$ and $0\le\rho\le\sqrt{p-1}$. Then
\[
 \norm{T_\rho g}_2\le\norm g_p
\]
for every complex-valued function on a Boolean cube.
\end{theorem}
The complex-valued statement follows from the real one by positivity:
$|T_\rho g|\le T_\rho|g|$ pointwise. We first isolate a fixed-level
consequence.

\begin{lemma}[Degree-bounded $L^1$--$L^2$ estimate]\label{lem:L1L2}
If $g$ has Fourier degree at most an integer $K\ge0$, then
\begin{equation}\label{eq:L1L2}
 \norm g_2\le\ee^K\norm g_1.
\end{equation}
Moreover, for $1\le j\le K$,
\begin{equation}\label{eq:fixed-level}
 \norm{g_j}_2\le\left(\frac{4\ee K}{j}\right)^{j/2}\norm g_1.
\end{equation}
\end{lemma}
\begin{proof}
The zero function and the case $K=0$ are immediate. Fix $1<p<2$ and put
$\rho=\sqrt{p-1}$. Since $g$ has degree at most $K$, Parseval gives
\[
 \norm{T_\rho g}_2\ge\rho^K\norm g_2.
\]
Theorem~\ref{thm:BB} therefore yields
\[
 \norm g_2\le(p-1)^{-K/2}\norm g_p.
\]
Let $\vartheta=(2-p)/p$. Interpolation between $L^1$ and $L^2$ gives
\[
 \norm g_p\le\norm g_1^{\vartheta}\norm g_2^{1-\vartheta}.
\]
After rearranging,
\[
 \frac{\norm g_2}{\norm g_1}\le(p-1)^{-K/(2\vartheta)}.
\]
As $p\uparrow2$,
\[
 -\frac{1}{2\vartheta}\log(p-1)
 =-\frac{p}{2(2-p)}\log(p-1)\longrightarrow1,
\]
which proves~\eqref{eq:L1L2}.

For the level estimate, let $0<s\le1$. Parseval and
Theorem~\ref{thm:BB} give
\[
 s^{j/2}\norm{g_j}_2\le\norm{T_{\sqrt s}g}_2\le\norm g_{1+s}.
\]
Interpolation and~\eqref{eq:L1L2} imply
\[
 \norm g_{1+s}
 \le\norm g_1^{(1-s)/(1+s)}\norm g_2^{2s/(1+s)}
 \le\exp\left(\frac{2Ks}{1+s}\right)\norm g_1.
\]
Choose $s=j/(4K)$. Then $s\le1/4$ and $2Ks/(1+s)\le j/2$, so
\[
 \norm{g_j}_2\le\left(\frac{4K}{j}\right)^{j/2}\ee^{j/2}\norm g_1,
\]
which is~\eqref{eq:fixed-level}.
\end{proof}

\begin{lemma}[First-level coefficient estimate]\label{lem:first}
There is an absolute constant $C_1$ such that every complex-valued $F$ of
degree at most $M\ge1$ satisfies
\begin{equation}\label{eq:first}
 A_1(F)\le C_1M\norm F_\infty.
\end{equation}
One may take $C_1=\pi/2$.
\end{lemma}
\begin{proof}
Fix $x\in\Om_n$ and set
\[
 p_x(t)=T_tF(x)=\sum_{r=0}^M t^rF_r(x),\qquad -1\le t\le1.
\]
Equation~\eqref{eq:noise-infty} and the identity
$T_tF(x)=T_{|t|}F(-x)$ for $t<0$ give
$\norm{p_x}_{L^\infty[-1,1]}\le\norm F_\infty$.
The trigonometric polynomial $g_x(\theta)=p_x(\cos\theta)$ has degree at
most $M$. Bernstein's inequality, valid also for complex trigonometric
polynomials~\cite[Theorem 3.8]{QZ}, gives
\[
 |p_x'(0)|=|g_x'(\pi/2)|\le M\norm F_\infty.
\]
Since $p_x'(0)=F_1(x)$,
\begin{equation}\label{eq:first-sup}
 \norm{F_1}_\infty\le M\norm F_\infty.
\end{equation}
Write $F_1(x)=\sum_j a_jx_j$. For each $t\in[0,2\pi]$, choose signs
$x_j(t)$ so that
$x_j(t)\operatorname{Re}(\ee^{-it}a_j)=|\operatorname{Re}(\ee^{-it}a_j)|$.
Then
\[
 \sum_j|\operatorname{Re}(\ee^{-it}a_j)|
 =\operatorname{Re}\bigl(\ee^{-it}F_1(x(t))\bigr)
 \le\norm{F_1}_\infty.
\]
Averaging in $t$ and using
\[
 \frac{1}{2\pi}\int_0^{2\pi}|\operatorname{Re}(\ee^{-it}a)|\,dt
 =\frac{2|a|}{\pi}
\]
yields $\sum_j|a_j|\le(\pi/2)\norm{F_1}_\infty$.
Combine this with~\eqref{eq:first-sup}.
\end{proof}

\section{Mixed norms and the first three levels}\label{sec:mixed}
For finite sets $I,J$ and a scalar array $a=(a_{ij})_{I\times J}$, use
\[
 \norm a_{\ell^p(I;\ell^q(J))}
 =\left[\sum_{i\in I}\left(\sum_{j\in J}|a_{ij}|^q\right)^{p/q}\right]^{1/p}.
\]
\begin{lemma}[Finite mixed-norm interpolation]\label{lem:interpolation}
Let $1\le p_0,p_1,q_0,q_1<\infty$ and $0\le\eta\le1$. If
\[
 \frac1p=\frac{1-\eta}{p_0}+\frac{\eta}{p_1},
 \qquad \frac1q=\frac{1-\eta}{q_0}+\frac{\eta}{q_1},
\]
then
\[
 \norm a_{\ell^p(I;\ell^q(J))}
 \le\norm a_{\ell^{p_0}(I;\ell^{q_0}(J))}^{1-\eta}
     \norm a_{\ell^{p_1}(I;\ell^{q_1}(J))}^{\eta}.
\]
\end{lemma}
\begin{proof}
The endpoints $\eta=0,1$ are immediate. Otherwise apply log-convexity of
the inner $\ell^q$ norm for each $i$, raise to the $p$th power, sum in $i$,
and use H\"older with conjugate exponents
$p_0/[p(1-\eta)]$ and $p_1/(p\eta)$.
\end{proof}

\begin{lemma}[Two-block coefficient inequality]\label{lem:two-block}
Let $d,e\ge1$, $r=d+e$, and let $a=(a_{A,B})$ be a finite scalar array.
Define
\[
 X_{d,e}=\left[\sum_A\left(\sum_B|a_{A,B}|^2\right)^{q_d/2}\right]^{1/q_d},
 \qquad
 Y_{d,e}=\left[\sum_B\left(\sum_A|a_{A,B}|^2\right)^{q_e/2}\right]^{1/q_e}.
\]
Then, with $\theta=d/r$,
\begin{equation}\label{eq:two-block}
 \left(\sum_{A,B}|a_{A,B}|^{q_r}\right)^{1/q_r}
 \le X_{d,e}^{\theta}Y_{d,e}^{1-\theta}.
\end{equation}
\end{lemma}
\begin{proof}
Apply Lemma~\ref{lem:interpolation} with
\[
 (p_0,q_0)=(q_d,2),\qquad (p_1,q_1)=(2,q_e),\qquad\eta=e/r.
\]
The identities
\[
 \frac1{q_r}=\frac dr\frac1{q_d}+\frac er\frac12
            =\frac dr\frac12+\frac er\frac1{q_e}
\]
show that both interpolated exponents equal $q_r$.
Finally, because $q_e\le2$, Minkowski gives
\[
 \norm a_{\ell^2(A;\ell^{q_e}(B))}
 \le\norm a_{\ell^{q_e}(B;\ell^2(A))}=Y_{d,e}.
\]
\end{proof}

We now establish the seed estimates that determine the $M$-part of the
weight. Bounds in Propositions~\ref{prop:second} and~\ref{prop:third} are
trivial when $M$ is below the relevant level.

\begin{proposition}[Second level]\label{prop:second}
There is an absolute constant $C_2$ such that every complex-valued $F$ of
degree at most $M\ge1$ satisfies
\begin{equation}\label{eq:second}
 A_2(F)\le C_2M^{3/2}\norm F_\infty.
\end{equation}
\end{proposition}
\begin{proof}
Assume $M\ge2$. Fix a partition $[n]=I\mathbin{\dot\cup}J$ and write
\[
 F(x,y)=\sum_{A\subseteq I}\sum_{B\subseteq J}a_{A,B}w_A(x)w_B(y),
 \qquad a_{A,B}=\wh F(A\cup B).
\]
For $i\in I$, define $B_i = B_i(y)$ to be the coefficient (depending on $y$) in front of $x_i$ in $F$, namely,
\[
 B_i(y)=\sum_{B\subseteq J}a_{\{i\},B}w_B(y).
\]
Its degree is at most $M-1$, and its first homogeneous part has
coefficients $a_{\{i\},\{j\}}$, $j\in J$. Apply Lemma~\ref{lem:L1L2} to $B_i$ with $y$ as variable and
$j=1$ in~\eqref{eq:fixed-level}, gives
\[
 \left(\sum_{j\in J}|a_{\{i\},\{j\}}|^2\right)^{1/2}
 \le C\sqrt M\norm{B_i}_1.
\]
Consequently,
\[
 \mathsf R:=\sum_{i\in I}
 \left(\sum_{j\in J}|a_{\{i\},\{j\}}|^2\right)^{1/2}
 \le C\sqrt M\,\E_y\sum_{i\in I}|B_i(y)|.
\]
For fixed $y$, the $B_i(y)$ are the first-level Fourier coefficients of
$x\mapsto F(x,y)$. By applying the first-level estimate Lemma~\ref{lem:first} to function $F$ with variable $x$,
\[
 \sum_{i\in I}|B_i(y)|\le C_1M\norm{F(\cdot,y)}_\infty
 \le C_1M\norm F_\infty.
\]
Thus
\begin{equation}\label{eq:rows2}
 \mathsf R\le CM^{3/2}\norm F_\infty.
\end{equation}
Interchanging $I$ and $J$ gives the column estimate
\begin{equation}\label{eq:cols2}
 \mathsf C:=\sum_{j\in J}
 \left(\sum_{i\in I}|a_{\{i\},\{j\}}|^2\right)^{1/2}
 \le CM^{3/2}\norm F_\infty.
\end{equation}
Lemma~\ref{lem:two-block} with $d=e=1$ gives
\begin{equation}\label{eq:cross2}
 \left(\sum_{i\in I,j\in J}|\wh F(\{i,j\})|^{4/3}\right)^{3/4}
 \le\mathsf R^{1/2}\mathsf C^{1/2}
 \le CM^{3/2}\norm F_\infty.
\end{equation}
Color each coordinate independently $I$ or $J$, with probability $1/2$.
Every two-element set is separated with probability $1/2$ and is then
counted exactly once in the sum. Taking expectations of the $4/3$-power
of~\eqref{eq:cross2} gives
\[
 \frac12A_2(F)^{4/3}\le C^{4/3}M^2\norm F_\infty^{4/3},
\]
which proves~\eqref{eq:second}.
\end{proof}

\begin{proposition}[Third level]\label{prop:third}
There is an absolute constant $C_3$ such that every complex-valued $F$ of
degree at most $M\ge1$ satisfies
\begin{equation}\label{eq:third}
 A_3(F)\le C_3M^{5/3}\norm F_\infty.
\end{equation}
\end{proposition}
\begin{proof}
Assume $M\ge3$. For a deterministic partition $[n]=I\mathbin{\dot\cup}J$,
consider the array
\[
 a_{i,T}=\wh F(\{i\}\cup T),\qquad i\in I,\quad T\subseteq J,\quad|T|=2.
\]
Lemma~\ref{lem:two-block} with $(d,e)=(1,2)$ gives
\begin{equation}\label{eq:third-interp}
 \norm a_{q_3}\le X^{1/3}Y^{2/3},
\end{equation}
where
\[
 X=\sum_{i\in I}\left(\sum_{|T|=2}|a_{i,T}|^2\right)^{1/2},
 \qquad
 Y=\left[\sum_{|T|=2}
       \left(\sum_{i\in I}|a_{i,T}|^2\right)^{q_2/2}\right]^{1/q_2}.
\]

For reader's convenience we repeat that for $i\in I$, define $B_i = B_i(y)$ to be the coefficient (depending on $y$) of $x_i$ in $F$, that is
\[
 B_i(y)=\sum_{B\subseteq J}a_{\{i\},B}w_B(y).
\]
Its degree is at most $M-1$, and its second homogeneous part has
coefficients $a_{\{i\},\{T\}}$, $j\in J$ where $|T|=2$. Apply Lemma~\ref{lem:L1L2} to $B_i$ with $y$ as variable and
$j=2$ in~\eqref{eq:fixed-level}, gives

\[
 \left(\sum_{|T|=2}|a_{i,T}|^2\right)^{1/2}\le CM\norm{B_i}_1,
\]

For fixed $y$, the $B_i(y)$ are the first-level Fourier coefficients of
$x\mapsto F(x,y)$. By applying the first-level estimate Lemma~\ref{lem:first} to function $F$ with variable $x$,

\begin{equation}\label{eq:third-X}
 X\le CM\E_y\sum_i|B_i(y)|
  \le CM \cdot \E_y [M \| F(\cdot,y) \|_\infty ]
 \le CM^2\norm F_\infty.
\end{equation}
For $T\subseteq J$, $|T|=2$, define $C_T = C_T(x)$ to be the coefficient (depending on $x$) in front  of $y^T$ in $F$:
\[
 C_T(x)=\sum_{A\subseteq I}a_{A,T}w_A(x).
\]
Its first homogeneous part has coefficients $a_{i,T}$. Since $q_2=4/3$
and $\rho_2=\sqrt{q_2-1}=1/\sqrt3$, applying Parseval and hypercontractivity bound Theorem~\ref{thm:BB} on $C_T$ with variable $x$ gives
\[
 \rho_2\left(\sum_{i\in I}|a_{i,T}|^2\right)^{1/2}
 \le\norm{T_{\rho_2}C_T}_2\le\norm{C_T}_{q_2}.
\]
Thus
\[
Y^{q_2}
=\left( \sum_{|T|=2} |a_{i,T}|^2 \right)^{q_2/2}
\le3^{q_2/2}\sum_{|T|=2}\norm{C_T}_{q_2}^{q_2}
 =3^{q_2/2}\E_x\sum_{|T|=2}|C_T(x)|^{q_2}.
\]
For fixed $x$, these $C_T(x)$ are the second-level Fourier coefficients of
$y\mapsto F(x,y)$. Applying the second-level estimate of Proposition~\ref{prop:second} to function $F$ with variable $y$, therefore gives
\begin{equation}\label{eq:third-Y}
 Y
 \le \sqrt{3} \left(\E_x\left[ \sum_{|T|=2} |C_T(x)|^{q_2} \right] \right)^{1/{q_2}
 \le \sqrt{3} \left(\E_x\left[ ( C M^{3/2} \|F(x,\cdot)\|_\infty )^{q_2} \right] \right)^{1/{q_2}} }
 \le CM^{3/2}\norm F_\infty.
\end{equation}
Combining~\eqref{eq:third-interp}--\eqref{eq:third-Y},
\[
 \norm a_{q_3}\le C(M^2)^{1/3}(M^{3/2})^{2/3}\norm F_\infty
 =CM^{5/3}\norm F_\infty.
\]
Under a random bicoloring, a three-element set has exactly one coordinate
in $I$ and two in $J$ with probability $3/8$. Averaging the $q_3$-powers
recovers $(3/8)A_3(F)^{q_3}$ on the left, proving the claim.
\end{proof}

\begin{remark} 
One may think that we can proceed with the same process and obtain a bound for $A_m$ looking like 
$$
A_m(F) \le C M^{5/3} \|F\|_{\infty}
$$
as in Corollary~\ref{cor:homogeneous} or a bound of the weighted square function as in Theorem~\ref{thm:weighted}
$$
\left(\sum_{r=1}^M\frac{A_r(F)^2}{w_{M,r}^2}\right)^{1/2}
 \le C \norm F_\infty
$$
by using each estimate for $A_m(F)$ and summing them up.

However, this naive idea fails since the constant $C$ depends on degree and blow up as degree becomes larger. An easy way to see it is that, at the end of both the second-level and third-level proof, we need to divide by the probability of splitting into given sizes (for second-level, it is $1+1$; while for third-level, it is $1+2$). As degree becomes larger, the probability of dividing into a single configuration will tend to zero even when we divide almost evenly.
\end{remark}

\section{The degree-adapted weighted square function}
\label{sec:weighted}
Fix $B>1/\log2$ and put
\begin{equation}\label{eq:eta}
 \eta_B=B\log2-1>0.
\end{equation}
Choose a fixed constant $C_B\ge8/\eta_B$. A sufficiently large integer
$R_B\ge6^4$, depending only on $B$, will be fixed in \sect{sec:high}.
For $M\ge1$, set
\begin{equation}\label{eq:cutoff}
 L_M=\left\lceil\max\{R_B,C_B\log(M+1)\}\right\rceil.
\end{equation}
Define
\begin{equation}\label{eq:mu}
 \mu_1=1,\qquad\mu_r=\left\lceil\frac{3r}{2}\right\rceil\quad(r\ge2),
\end{equation}
\begin{equation}\label{eq:sigma}
 \sigma_{M,r}=\min\{r,L_M\}^8,
\end{equation}
and
\begin{equation}\label{eq:weight}
 w_{M,r}=M^{\mu_r/r}r^B\sigma_{M,r}\qquad(1\le r\le M).
\end{equation}
Here $\mu_r/r$ is the actual exponent of $M$. The purpose of recording the
numerator $\mu_r$ separately is that the interpolation weights $d/r$ and
$e/r$ turn its splitting cost into
\[
 \frac{(M^{\mu_d/d})^{d/r}(M^{\mu_e/e})^{e/r}}{M^{\mu_r/r}}
 =M^{(\mu_d+\mu_e-\mu_r)/r}.
\]
Thus ordinary additivity of $\mu$ is exactly the relevant arithmetic.

\begin{remark}
One will see that we need to make the factor
$$
\frac{w_{M,d}^{\theta}w_{M,e}^{1-\theta}}{w_{M,r}}
$$
as small as possible in both section \ref{sec:abstract} and subsection \ref{sec:low}. In particular, for small $r$ and for $r$ divided into $d+e$, we must have
$$
\mu_r \ge \mu_d + \mu_e
$$
since otherwise, we will have a factor at least $M^{1/r}$ blowing up as $M$ becomes larger. Moreover, $\mu_2 = \frac{3}{2} \times 2 = 3$ and $\mu_3 = \frac{5}{3} \times 3 = 5$ from the second-level and third-level estimates, which we call ``seeds."

One may think taking $\mu_r = 3r/2$. But this fails for $\mu_3$. One may then think about taking $\mu_r = 5r/3$. Indeed, this works, but we now explain why it can be made better to $\lceil 3r/2 \rceil$. Suppose we are dividing $r$ into $d+e$ and then divide both $d$ and $e$ and so on until they turn into $2$'s and $3's$. Assume $r$ divided into $a$ $2$'s and $b$ $3$'s. We have
$$
\begin{aligned}
    r &= 2a+3b \\
    \mu_r &\ge \mu_d + \mu_e \ge \cdots \ge a \mu_2 + b \mu_3 = 3a+5b.
\end{aligned}
$$
This implies that $\mu_r \ge 3r/2 + b/2$. If we want to make it as small as possible, we should manage our division so that $b$, the number of $3$, is as fewest as possible. This is made plausible by our design of splitting window in \eqref{eq:Dlow} in subsection \ref{sec:low} which guarantees that we have $b=0$ for even $r$ and $b=1$ for odd $r$. Therefore, $\mu_r$ is exactly equal to $\lceil 3r/2 \rceil$.

\begin{figure}[htbp]
\centering
\begin{tikzcd}[bh tree, baseline=(general-root.center)]
  & & & & & & &
  |[bh above=r, alias=general-root]| {}
  & & & & & & & \\[5mm]
  & & & |[bh above=d]| {}
  & & & & & & & & |[bh above=e]| {}
  & & & \\[2mm]
  & |[bh vertex]| {}
  & & & & |[bh vertex]| {}
  & & & & |[bh vertex]| {}
  & & & & |[bh vertex]| {} & \\[-1mm]
  & |[yshift=1mm]| \vdots
  & & & & |[yshift=1mm]| \vdots
  & & & & |[yshift=1mm]| \vdots
  & & & & |[yshift=1mm]| \vdots & \\[-1mm]
  & |[bh vertex]| {}
  & & & & |[bh vertex]| {}
  & & & & |[bh vertex]| {}
  & & & & |[bh vertex]| {} & \\
  |[bh below=2]| {} & & |[bh below=2]| {} & &
  |[bh below=2]| {} & & |[bh below=3]| {} & &
  |[bh below=2]| {} & & |[bh below=3]| {} & &
  |[bh below=2]| {} & & |[bh below=2]| {}
  \arrow[from=1-8, to=2-4]
  \arrow[from=1-8, to=2-12]
  \arrow[from=2-4, to=3-2]
  \arrow[from=2-4, to=3-6]
  \arrow[from=2-12, to=3-10]
  \arrow[from=2-12, to=3-14]
  \arrow[from=5-2, to=6-1]
  \arrow[from=5-2, to=6-3]
  \arrow[from=5-6, to=6-5]
  \arrow[from=5-6, to=6-7]
  \arrow[from=5-10, to=6-9]
  \arrow[from=5-10, to=6-11]
  \arrow[from=5-14, to=6-13]
  \arrow[from=5-14, to=6-15]
\end{tikzcd}
\hspace{10mm}
\begin{tikzcd}[
  bh tree,
  row sep={11mm,between origins},
  baseline=(example-root.center)
]
  & & & & &
  |[bh above=\ExampleRoot, alias=example-root]| {}
  & & & & \\
  & & |[bh above=6]| {}
  & & & & & & |[bh above=5]| {} & \\
  |[bh below=2]| {}
  & & & & |[bh above=4]| {}
  & & & |[bh below=2]| {} & & |[bh below=3]| {} \\
  & & & |[bh below=2]| {}
  & & |[bh below=2]| {} & & & &
  \arrow[from=1-6, to=2-3]
  \arrow[from=1-6, to=2-9]
  \arrow[from=2-3, to=3-1]
  \arrow[from=2-3, to=3-5]
  \arrow[from=2-9, to=3-8]
  \arrow[from=2-9, to=3-10]
  \arrow[from=3-5, to=4-4]
  \arrow[from=3-5, to=4-6]
\end{tikzcd}
\caption{A general recursive splitting tree and a concrete example of $r=13$.}
\label{fig:splitting-trees}
\end{figure}
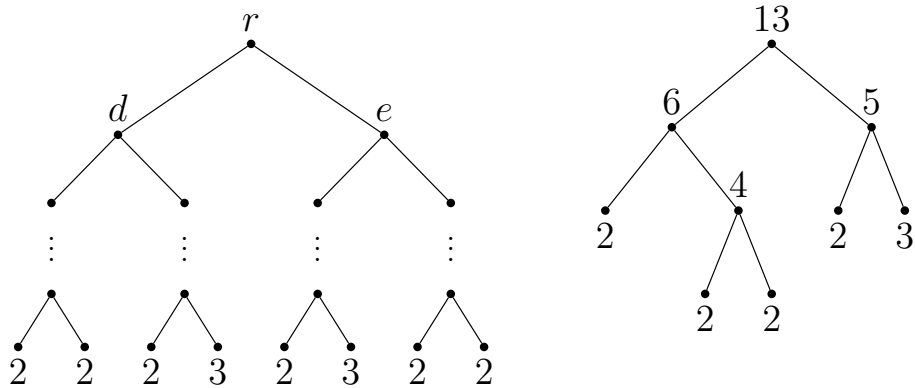
\end{remark}

For integers $M,N\ge1$, let
\begin{equation}\label{eq:Gamma}
 \Gamma_{M,N}(B)
 =\sup_{\substack{1\le n\le N,\ 0\ne F:\Om_n\to\C\\\deg F\le M}}
 \frac{\left(\sum_{r=1}^M A_r(F)^2/w_{M,r}^2\right)^{1/2}}
      {\norm F_\infty}.
\end{equation}
This is finite before any uniform estimate is proved. Indeed,
$|\wh F(A)|\le\norm F_\infty$ and, for $r\le N$,
$A_r(F)\le\binom{N}{r}^{1/q_r}\norm F_\infty$.
The finite dimension cap avoids assuming the conclusion during the
bootstrap. 

Let
\[
 Q(x,y)=\sum_{d,e\ge0}Q_{d,e}(x,y)
\]
be a polynomial of total degree at most $M$ on
$\Om_{n_x}\times\Om_{n_y}$, where $0\le n_x,n_y\le N$, and write
\[
 Q_{d,e}(x,y)=\sum_{|A|=d,|B|=e}a^{d,e}_{A,B}w_A(x)w_B(y).
\]
The mixed-norm formulas of Lemma~\ref{lem:two-block} define $X_{d,e}$
whenever $d\ge1,e\ge0$, and $Y_{d,e}$ whenever $d\ge0,e\ge1$; when one
block degree is zero its index set consists of the empty set. Coefficients
with $d+e>M$ are zero. For $d\ge2,e\ge0$, put
\[
 \rho_d=\sqrt{\frac{d-1}{d+1}},\qquad\wtX_{d,e}=\rho_d^eX_{d,e},
\]
and, for $d\ge0,e\ge2$, put $\wtY_{d,e}=\rho_e^dY_{d,e}$.

\begin{lemma}[Weighted simultaneous row and column control]\label{lem:simultaneous}
For every $M,N$ and every $Q$ as above,
\begin{equation}\label{eq:rows}
 \sum_{\substack{d\ge2,e\ge0\\d+e\le M}}
 \frac{\wtX_{d,e}^2}{w_{M,d}^2}
 \le\Gamma_{M,N}(B)^2\norm Q_\infty^2,
\end{equation}
\begin{equation}\label{eq:cols}
 \sum_{\substack{d\ge0,e\ge2\\d+e\le M}}
 \frac{\wtY_{d,e}^2}{w_{M,e}^2}
 \le\Gamma_{M,N}(B)^2\norm Q_\infty^2.
\end{equation}
\end{lemma}
\begin{proof}
We prove~\eqref{eq:rows}. Fix $2\le d\le M$. For $|A|=d$, let
\[
 B_{d,A}(y)=\sum_{e\ge0}\sum_{|B|=e}a^{d,e}_{A,B}w_B(y).
\]
Since $q_d\le2$, Minkowski interchanges the
$\ell_e^2(\ell_A^{q_d})$ norm in the required direction. Parseval in $y$
and Theorem~\ref{thm:BB}, with $\rho_d=\sqrt{q_d-1}$, then give
\begin{align*}
 \left(\sum_{e\ge0}\wtX_{d,e}^2\right)^{1/2}
 &\le\left[\sum_{|A|=d}
   \left(\sum_{e\ge0}\sum_{|B|=e}\rho_d^{2e}|a^{d,e}_{A,B}|^2\right)^{q_d/2}
   \right]^{1/q_d}\\
 &=\left[\sum_{|A|=d}\norm{T_{\rho_d}B_{d,A}}_2^{q_d}\right]^{1/q_d}\\
 &\le\left[\sum_{|A|=d}\norm{B_{d,A}}_{q_d}^{q_d}\right]^{1/q_d}.
\end{align*}
Set
\[
 g_d(y)=\left(\sum_{|A|=d}|B_{d,A}(y)|^{q_d}\right)^{1/q_d}.
\]
The last expression equals $\norm{g_d}_{q_d}\le\norm{g_d}_2$, because the
measure is normalized and $q_d\le2$. Divide by $w_{M,d}$, square, sum in
$d$, and use Tonelli:
\[
 \sum_{\substack{d\ge2,e\ge0\\d+e\le M}}
 \frac{\wtX_{d,e}^2}{w_{M,d}^2}
 \le\E_y\sum_{d=2}^M\frac{g_d(y)^2}{w_{M,d}^2}.
\]
For each fixed $y$, $g_d(y)=A_d(Q(\cdot,y))$. The slice has degree at most
$M$ in at most $N$ variables, and its supremum norm is at most
$\norm Q_\infty$. Definition~\eqref{eq:Gamma} therefore bounds the last
display by $\Gamma_{M,N}(B)^2\norm Q_\infty^2$. A zero slice or an empty
$x$-block contributes zero. Exchanging $x$ and $y$ proves~\eqref{eq:cols}.
\end{proof}

For $d,e\ge2$, $r=d+e$, and $\theta=d/r$, define the exact Boolean
deweighting factor
\begin{equation}\label{eq:deweight}
 R(d,e)=\rho_d^{-ed/r}\rho_e^{-de/r}.
\end{equation}
Writing $a(Q_{d,e})$ for the coefficient vector,
Lemma~\ref{lem:two-block} gives
\begin{equation}\label{eq:deweighted-block}
 \norm{a(Q_{d,e})}_{q_r}
 \le R(d,e)\wtX_{d,e}^{\theta}\wtY_{d,e}^{1-\theta}.
\end{equation}
We shall use the uniform bound
\begin{equation}\label{eq:R-three}
 R(d,e)\le3.
\end{equation}
To verify it, put $\phi(t)=t\log((t+1)/(t-1))$ for $t\ge2$. Since
\[
 \phi'(t)=\log\frac{t+1}{t-1}-\frac{2t}{t^2-1}\le0,
\]
we have $\phi(t)\le\phi(2)=2\log3$. The derivative inequality follows from
\[
 \operatorname{arctanh}u=\int_0^u\frac{ds}{1-s^2}
 \le\frac{u}{1-u^2}\qquad(0<u<1).
\]
Finally,
\[
 \log R(d,e)=\frac{e\phi(d)+d\phi(e)}{2(d+e)}\le\log3,
\]
which proves~\eqref{eq:R-three}.

\section{An abstract random splitting}\label{sec:abstract}
For each $r\ge4$, let $D_r$ be a nonempty subset of $\{2,\ldots,r-2\}$,
symmetric under $d\mapsto r-d$. Put
\begin{equation}\label{eq:capture-def}
 p_r=2^{-r}\sum_{d\in D_r}\binom rd,\qquad k_r=|D_r|.
\end{equation}
Color each coordinate independently $x$ or $y$, with probability $1/2$,
and let $Q_\omega$ be the resulting two-block version of $F$.
This only relabels coordinates, so $\norm{Q_\omega}_\infty=\norm F_\infty$.
Each Walsh monomial has a unique pair of colored index sets. Thus
\begin{equation}\label{eq:ancestry}
 \E_\omega\sum_{d\in D_r}
 \norm{a(Q_{\omega,d,r-d})}_{q_r}^{q_r}
 =p_rA_r(F)^{q_r}.
\end{equation}
For $d\in D_r$, $e=r-d$, and $\theta=d/r$, define
\begin{equation}\label{eq:KM}
 K_M(r,d)=p_r^{-1/q_r}k_r^{1/(2r)}R(d,e)
 \frac{w_{M,d}^{\theta}w_{M,e}^{1-\theta}}{w_{M,r}}.
\end{equation}

\begin{remark}
In the above equality, note that the splitting window $D_r$, whose size is $k_r$, will decide the probability $p_r$. We want to make $K_M(r,d)$ and the ratio of weighted function 
$$
\frac{w_{M,d}^{\theta}w_{M,e}^{1-\theta}}{w_{M,r}}
$$
in \eqref{eq:weight-ratio} to be as small as possible. The most important thing for us is to design the weight $w_{M,r}$.

\end{remark}

\begin{proposition}[Abstract weighted contraction]\label{prop:abstract}
Assume the sets $D_r$ have been chosen for $4\le r\le M$. Define
\begin{equation}\label{eq:c-def}
 c^2=2\sup_{4\le r\le M}\ \sup_{d\in D_r}
 \max\left\{\frac dr,1-\frac dr\right\}K_M(r,d)^2.
\end{equation}
Then every
$F:\Om_n\to\C$, $1\le n\le N$, of degree at most $M$ satisfies
\begin{equation}\label{eq:abstract-contraction}
 \left(\sum_{r=4}^M\frac{A_r(F)^2}{w_{M,r}^2}\right)^{1/2}
 \le c\,\Gamma_{M,N}(B)\norm F_\infty.
\end{equation}
\end{proposition}
\begin{proof}
For each coloring $\omega$, set
\[
 C_r(\omega)=\left(\sum_{d\in D_r}
 \norm{a(Q_{\omega,d,r-d})}_{q_r}^{q_r}\right)^{1/q_r}.
\]
By~\eqref{eq:ancestry} and monotonicity of probability-space $L^p$ norms,
\[
 A_r(F)=p_r^{-1/q_r}\bigl(\E_\omega C_r(\omega)^{q_r}\bigr)^{1/q_r}
 \le p_r^{-1/q_r}\bigl(\E_\omega C_r(\omega)^2\bigr)^{1/2}.
\]
Since $1/q_r-1/2=1/(2r)$, the finite-dimensional coefficient norm inequality
also gives
\[
 C_r(\omega)\le k_r^{1/(2r)}
 \left(\sum_{d\in D_r}\norm{a(Q_{\omega,d,r-d})}_{q_r}^2\right)^{1/2}.
\]
Using~\eqref{eq:deweighted-block} and inserting the weights, we obtain
\begin{equation}\label{eq:main-double-sum}
 \sum_{r=4}^M\frac{A_r(F)^2}{w_{M,r}^2}
 \le\E_\omega\sum_{r=4}^M\sum_{d\in D_r}
 K_M(r,d)^2U_{d,e}^{2\theta}V_{d,e}^{2(1-\theta)},
\end{equation}
where, for this coloring,
\[
 U_{d,e}=\frac{\wtX_{d,e}}{w_{M,d}},\qquad
 V_{d,e}=\frac{\wtY_{d,e}}{w_{M,e}},\qquad e=r-d,\quad\theta=d/r.
\]
Since $\Gamma_{M,N}(B)>0$ (test the function $x\mapsto x_1$), put
$G=\Gamma_{M,N}(B)\norm F_\infty>0$. Lemma~\ref{lem:simultaneous}
gives, in particular, the budgets over all retained pairs:
\[
 \sum_{r=4}^M\sum_{d\in D_r}U_{d,r-d}^2\le G^2,
 \qquad
 \sum_{r=4}^M\sum_{d\in D_r}V_{d,r-d}^2\le G^2.
\]
Normalize $u_{d,e}=U_{d,e}^2/G^2$ and $v_{d,e}=V_{d,e}^2/G^2$.
Weighted AM--GM gives
\[
 u_{d,e}^{\theta}v_{d,e}^{1-\theta}
 \le\theta u_{d,e}+(1-\theta)v_{d,e}.
\]
For every retained pair, definition~\eqref{eq:c-def} implies both
$K_M(r,d)^2\theta\le c^2/2$ and
$K_M(r,d)^2(1-\theta)\le c^2/2$. Hence the double sum
in~\eqref{eq:main-double-sum}, divided by $G^2$, is at most
\[
 \frac{c^2}{2}\sum_{r,d}u_{d,e}
 +\frac{c^2}{2}\sum_{r,d}v_{d,e}\le c^2.
\]
Taking expectations and square roots proves~\eqref{eq:abstract-contraction}.
\end{proof}

The exact weight ratio in~\eqref{eq:KM} is
\begin{equation}\label{eq:weight-ratio}
 \frac{w_{M,d}^{\theta}w_{M,e}^{1-\theta}}{w_{M,r}}
 =M^{(\mu_d+\mu_e-\mu_r)/r}\ee^{-Bh(\theta)}
  \frac{\sigma_{M,d}^{\theta}\sigma_{M,e}^{1-\theta}}{\sigma_{M,r}},
\end{equation}
where
\begin{equation}\label{eq:entropy}
 h(t)=-t\log t-(1-t)\log(1-t).
\end{equation}

\section{Choice of the weight and splitting windows}\label{sec:choices}
We now choose $R_B$ and $D_r$ so that the constant in
Proposition~\ref{prop:abstract} is uniformly smaller than one. To make the
two ranges and their combination precise, define the \emph{local multiplier}
\[
 \kappa_M(r,d)
 :=\sqrt{2\max\{d/r,1-d/r\}}\,K_M(r,d).
\]
The sets $D_r$ may depend on $M$ through $L_M$; this dependence is
suppressed in the notation. For the choices made below, let
\[
 \begin{aligned}
 \Glow(M)&:=\sup_{\substack{4\le r\le M,\ r<L_M\\d\in D_r}}
                   \kappa_M(r,d),\\
 \Ghigh(M)&:=\sup_{\substack{4\le r\le M,\ r\ge L_M\\d\in D_r}}
                   \kappa_M(r,d).
 \end{aligned}
\]
A supremum over an empty range is zero. Since these two ranges partition
all retained splits, the constant $c=c(M)$ of~\eqref{eq:c-def} satisfies
\[
 \boxed{c(M)=\max\{\Glow(M),\Ghigh(M)\}.}
\]
This identity concerns the supremum of the \emph{same local multiplier}.
It is not obtained by adding two separately estimated square functions.

\subsection{Arithmetic of the low-level exponent}\label{sec:arithmetic}
\begin{lemma}[Additivity defect of $\mu_r$]\label{lem:defect}
Let $d,e\ge2$ and $r=d+e$. Then
\begin{equation}\label{eq:defect}
 \mu_d+\mu_e-\mu_r=
 \begin{cases}
  1,&r\text{ is even and }d,e\text{ are odd},\\
  0,&\text{otherwise}.
 \end{cases}
\end{equation}
\end{lemma}
\begin{proof}
For $s\ge2$, write $\mu_s=(3s+\epsilon_s)/2$, where $\epsilon_s=0$
for even $s$ and $\epsilon_s=1$ for odd $s$. The defect equals
$(\epsilon_d+\epsilon_e-\epsilon_r)/2$, which gives~\eqref{eq:defect}.
\end{proof}

\subsection{The low and intermediate range}\label{sec:low}
For $4\le r<L_M$, choose
\begin{equation}\label{eq:Dlow}
 \Dlow=\left\{d\in\{2,\ldots,r-2\}:\frac r4\le d\le\frac{3r}{4},
          \quad d\text{ even if }r\text{ is even}\right\}.
\end{equation}
For odd $r$ there is no parity restriction. This set is symmetric under
$d\mapsto r-d$. 

\begin{remark}
One may think about the reason why we do not need to impose a parity condition when $r$ is odd. The reason is that if $r$ is odd and $r=d+e$, then $\mu_r = \mu_d+\mu_e$ where $\mu_r$ is defined to be $\lceil 3r/2 \rceil$. However, when $r$ is even, $\mu_d + \mu_e -\mu_r = 1$ when both $d,e$ are odd. Therefore, we impose the parity condition to prevent this from happening.

Basically, the mechanism behind is to divide a large $r$ into smaller numbers, which will eventually go to either $2$ or $3$. Moreover, we want $2$ to be as many as possible. Therefore, we impose a parity condition on even $r$ so that it could only divide into even numbers $d,e$. For odd $r$, each division must divide $r$ into an odd number and an even number. The even number will eventually divides into $2$'s since the parity condition imposed on even numbers. For the new odd number, it is again divided into another odd and another even number. Therefore, when we repeat the process, there will eventually be exactly one $3$ and the others must be $2$. One can look at Figure \ref{fig:splitting-trees} in section \ref{sec:weighted} to get a better feeling.
\end{remark}


\begin{lemma}[Low-range probability]\label{lem:low-capture}
For every integer $r\ge4$, the probability in~\eqref{eq:capture-def}, with
$D_r=\Dlow$, satisfies
\begin{equation}\label{eq:low-prob}
 p_r\ge\frac14.
\end{equation}
\end{lemma}
\begin{proof}
Let $Z_r\sim\operatorname{Bin}(r,1/2)$. Its variance is $r/4$, so
Chebyshev's inequality gives
\[
 \Prob\{|Z_r-r/2|>r/4\}\le\frac4r.
\]
For $r\ge16$ the central window is contained in $\{2,\ldots,r-2\}$.
If $r\ge16$ is odd, the preceding estimate yields $p_r\ge1-4/r\ge3/4$.
If $r\ge16$ is even, $\Prob\{Z_r\text{ is even}\}=1/2$, and hence
\[
 p_r\ge\frac12-\frac4r\ge\frac14.
\]
The remaining values $4\le r\le15$ are checked exactly in
Appendix~\ref{app:capture}.
\end{proof}

For $d\in\Dlow$, Lemma~\ref{lem:defect} gives
$\mu_d+\mu_e=\mu_r$. Since $d,e,r<L_M$,
\[
 \frac{\sigma_{M,d}^{\theta}\sigma_{M,e}^{1-\theta}}{\sigma_{M,r}}
 =\frac{d^{8\theta}e^{8(1-\theta)}}{r^8}=\ee^{-8h(\theta)}.
\]
Thus~\eqref{eq:weight-ratio} becomes
\begin{equation}\label{eq:low-ratio}
 \frac{w_{M,d}^{\theta}w_{M,e}^{1-\theta}}{w_{M,r}}
 =\ee^{-(B+8)h(\theta)}.
\end{equation}

\begin{lemma}[Strict low-range contraction]\label{lem:low-contract}
Set $c_*:=45/(2\ee^4)<1$. For every $M\ge1$, with $D_r=\Dlow$ on the
low range,
\[
 \Glow(M)\le c_*.
\]
\end{lemma}
\begin{proof}
The assertion is immediate for an empty range. Otherwise
$\theta\in[1/4,3/4]$, $p_r\ge1/4$, and $k_r\le r+1$.
Since $q_r\ge q_4=8/5$ and $(r+1)^{1/(2r)}$ decreases for $r\ge1$,
\[
 p_r^{-1/q_r}\le4^{5/8},\qquad k_r^{1/(2r)}\le5^{1/8}.
\]
Using~\eqref{eq:R-three}, \eqref{eq:low-ratio}, and
$h(\theta)\ge h(1/4)$ gives the pointwise bound
\[
 \kappa_M(r,d)
 \le\sqrt{\frac32}\,3\cdot4^{5/8}\cdot5^{1/8}
       \ee^{-(B+8)h(1/4)}.
\]
The inequalities
\[
 h(1/4)>\frac12,\quad 4^{5/8}<3,\quad 5^{1/8}<2,\quad
 \sqrt{3/2}<\frac54,
\]
together with $B>0$, make the right side strictly less than
$45/(2\ee^4)=c_*$. Also
$\ee^4>1+4+4^2/2+4^3/6>45/2$, so $c_*<1$.
Taking the low-range supremum proves the result.
\end{proof}

\subsection{The high range}\label{sec:high}
For $r\ge L_M$, choose the shrinking balanced window
\begin{equation}\label{eq:Dhigh}
 \Dhigh=\left\{d\in\{2,\ldots,r-2\}:
           |d-r/2|\le r^{3/4}\right\}.
\end{equation}
No parity restriction is imposed. As before, let $e=r-d$ and $\theta=d/r$.

\begin{lemma}[High-range elementary bounds]\label{lem:high-bounds}
For every integer $r\ge6^4$ and every $d\in\Dhigh$, the corresponding probability and cardinality satisfy
\begin{equation}\label{eq:high-prob}
 p_r\ge1-2\ee^{-2\sqrt r},
\end{equation}
\begin{equation}\label{eq:high-card}
 k_r\le3r^{3/4},
\end{equation}
\begin{equation}\label{eq:balance}
 \max\{\theta,1-\theta\}\le\frac12+r^{-1/4},
\end{equation}
\begin{equation}\label{eq:entropy-high}
 h(\theta)\ge\log2-\frac3{\sqrt r},
\end{equation}
\begin{equation}\label{eq:R-high}
 R(d,e)\le\exp\left(1+\frac{36}{r^2}\right).
\end{equation}
\end{lemma}
\begin{proof}
For $r\ge6^4$, the interval $[r/2-r^{3/4},r/2+r^{3/4}]$ lies in
$[r/3,2r/3]\subseteq[2,r-2]$; in particular, there is no extra truncation
of the binomial window by the block-degree restriction.
For $Z_r\sim\operatorname{Bin}(r,1/2)$,
\[
 \Prob\{|Z_r-r/2|>t\}\le2\exp(-2t^2/r).
\]
For completeness, this follows from
$\E\exp(s(Z_r-r/2))=(\cosh(s/2))^r\le\exp(rs^2/8)$,
exponential Markov inequality, and the choice $s=4t/r$ for each tail.
Taking $t=r^{3/4}$ proves~\eqref{eq:high-prob}. There are at most
$2r^{3/4}+1\le3r^{3/4}$ integers in the interval, proving
\eqref{eq:high-card}; \eqref{eq:balance} follows by dividing by $r$.

We have $|\theta-1/2|\le r^{-1/4}\le1/6$, so $\theta\in[1/3,2/3]$.
Since $h(1/2)=\log2$, $h'(1/2)=0$, and
\[
 |h''(t)|=\frac1{t(1-t)}\le\frac92\qquad(1/3\le t\le2/3),
\]
Taylor's theorem gives
\[
 h(\theta)\ge\log2-\frac94(\theta-1/2)^2
 \ge\log2-\frac{9}{4\sqrt r}
 \ge\log2-\frac3{\sqrt r}.
\]
This is~\eqref{eq:entropy-high}.

Finally,
\[
 \log R(d,e)=\frac{de}{2r}
 \left(\log\frac{d+1}{d-1}+\log\frac{e+1}{e-1}\right).
\]
For $t\ge2$,
$\log((t+1)/(t-1))=2\operatorname{arctanh}(1/t)\le2t/(t^2-1)$.
Consequently,
\begin{align*}
 \log R(d,e)
 &\le\frac{de}{r}\left(\frac{d}{d^2-1}+\frac{e}{e^2-1}\right)\\
 &=1+\frac{e}{r(d^2-1)}+\frac{d}{r(e^2-1)}.
\end{align*}
Since $d,e\ge r/3$ and $d,e\ge2$, each of the last two terms is at most
$18/r^2$: for example,
\[
 \frac{e}{r(d^2-1)}\le\frac1{d^2-1}\le\frac2{d^2}\le\frac{18}{r^2}.
\]
Exponentiating proves~\eqref{eq:R-high}.
\end{proof}

For $r\ge L_M$, Lemma~\ref{lem:defect} gives $\mu_d+\mu_e-\mu_r\le1$.
Also
\[
 \frac{\sigma_{M,d}^{\theta}\sigma_{M,e}^{1-\theta}}{\sigma_{M,r}}\le1,
\]
because $\sigma_{M,r}=L_M^8$ and $\sigma_{M,d},\sigma_{M,e}\le L_M^8$.
Moreover,
\begin{equation}\label{eq:defect-control}
 M^{1/r}\le
 \exp\left(\frac{\log M}{C_B\log(M+1)}\right)
 \le\ee^{1/C_B}.
\end{equation}
This is a uniform bound, not a claim that $M^{1/r}\to1$ throughout the
whole high range with $C_B$ fixed.

Define
\begin{equation}\label{eq:EB}
 \begin{split}
 E_B(r)={}&\frac12\log(1+2r^{-1/4})
       -\log(1-2\ee^{-2\sqrt r})\\
       &+\frac{\log(3r^{3/4})}{2r}
        +\frac{36}{r^2}+\frac{3B}{\sqrt r}.
 \end{split}
\end{equation}
For $r\ge6^4$ this is well-defined, and $E_B(r)\to0$ as $r\to\infty$. We now fix an integer $R_B\ge6^4$ so large that
\begin{equation}\label{eq:RB}
 E_B(r)\le\frac{\eta_B}{4}\qquad\text{for every integer }r\ge R_B.
\end{equation}
This choice depends only on $B$ and completes the definition of $L_M$.

\begin{lemma}[Strict high-range contraction]\label{lem:high-contract}
For every $M\ge1$, with $D_r=\Dhigh$ on the high range,
\begin{equation}\label{eq:high-contract}
 \Ghigh(M)\le\exp(-5\eta_B/8)<1.
\end{equation}
\end{lemma}
\begin{proof}
There is nothing to prove for an empty high range. For every retained
high-range pair, \eqref{eq:weight-ratio} and~\eqref{eq:defect-control} give
\[
 \frac{w_{M,d}^{\theta}w_{M,e}^{1-\theta}}{w_{M,r}}
 \le\ee^{1/C_B}\ee^{-Bh(\theta)}.
\]
Taking the logarithm of the full local multiplier, and using
Lemma~\ref{lem:high-bounds}, yields
\begin{align*}
 \log\kappa_M(r,d)
 ={}&\tfrac12\log\bigl(2\max\{\theta,1-\theta\}\bigr)
      -\tfrac1{q_r}\log p_r+\tfrac1{2r}\log k_r\\
    &+\log R(d,e)
      +\log\frac{w_{M,d}^{\theta}w_{M,e}^{1-\theta}}{w_{M,r}}\\
 \le{}&1-B\log2+\frac1{C_B}+E_B(r).
\end{align*}
Here $-q_r^{-1}\log p_r\le-\log p_r$ because $0<p_r\le1$ and $q_r\ge1$.
Since $r\ge L_M\ge R_B$, equations~\eqref{eq:eta} and~\eqref{eq:RB},
together with $C_B\ge8/\eta_B$, give
\[
 \log\kappa_M(r,d)
 \le-\eta_B+\frac{\eta_B}{8}+\frac{\eta_B}{4}
 =-\frac{5\eta_B}{8}.
\]
Exponentiating and taking the high-range supremum proves the claim.
\end{proof}

\begin{remark}
One may think of why not imposing a parity condition on the splitting window also for $r \ge L_M$ that is even. There are two reason behind. First, if we do this, then $p_r$ is roughly $1/2$ and $E_B(r)$ defined in \eqref{eq:EB} will tend to $\log 2$ as $r \to \infty$. Although the $M^{1/r}$ and its consequent $1/C_B$ will vanish, we will obtain:
$$
\log\kappa_M(r,d) \le 1 - B \log 2 + E_B(r).
$$
Since $E_B(r) \to 1/2$ as $r$ tends to infinity, in order for $\log\kappa_M(r,d)$ to be negative, we now need
$$
B \log 2 > 1+ \frac{1}{2} = \frac{3}{2} \ \Longleftrightarrow \ B > \frac{3}{2 \log 2},
$$
which makes our exponent to be bigger $3/2+3/(2 \log 2)$.

On the other hand, if we don't impose a parity condition, $M^{1/r}$ will appear. However, it is easy to control when $r$ is above the scale $L_M = \log(M+1)$. This is the main reason why we handle low and high ranges separately.

\end{remark}

\subsection{Combining the two ranges}\label{sec:combine}
For each $4\le r\le M$, use $D_r=\Dlow$ if $r<L_M$, and $D_r=\Dhigh$
if $r\ge L_M$. These sets are nonempty and symmetric. Define
\[
 c_B:=\max\left\{\frac{45}{2\ee^4},\exp\left(-\frac{5\eta_B}{8}\right)\right\}<1.
\]
Lemmas~\ref{lem:low-contract} and~\ref{lem:high-contract} give
\[
 c(M)=\max\{\Glow(M),\Ghigh(M)\}\le c_B
 \qquad\text{for every }M\ge1.
\]
The bound is independent of $M$ and $N$. Crucially, we apply
Proposition~\ref{prop:abstract} \emph{once}, using the union of the two
ranges. Its two energy budgets run over all retained pairs simultaneously.
No factor $\sqrt{\Glow(M)^2+\Ghigh(M)^2}$, and no sum of the two
contraction constants, is introduced.

\begin{proposition}[Uniform high-level contraction]\label{prop:uniform}
For every $B>1/\log2$, the constant $c_B<1$ defined above satisfies,
for every $M,N\ge1$ and every $F:\Om_n\to\C$ with $1\le n\le N$ and
$\deg F\le M$,
\begin{equation}\label{eq:uniform-contraction}
 \left(\sum_{r=4}^M\frac{A_r(F)^2}{w_{M,r}^2}\right)^{1/2}
 \le c_B\Gamma_{M,N}(B)\norm F_\infty.
\end{equation}
\end{proposition}
\begin{proof}
Apply Proposition~\ref{prop:abstract} with the preceding choice of $D_r$
and the bound $c(M)\le c_B$. The case $M<4$ has a zero left side.
\end{proof}

\section{Closure of the weighted bootstrap}\label{sec:closure}
\begin{theorem}[Uniform degree-adapted weighted estimate]\label{thm:weighted}
For every $B>1/\log2$, there is $K_B<\infty$ such that every complex-valued
$F:\Om_n\to\C$ of degree at most $M$ satisfies
\begin{equation}\label{eq:weighted-theorem}
 \left(\sum_{r=1}^M\frac{A_r(F)^2}{w_{M,r}^2}\right)^{1/2}
 \le K_B\norm F_\infty,
\end{equation}
uniformly in $M$ and $n$.
\end{theorem}
\begin{proof}
Because $L_M\ge R_B>3$, the first three weights, whenever their levels
are present, are
\[
 w_{M,1}=M,\qquad w_{M,2}=M^{3/2}2^{B+8},\qquad
 w_{M,3}=M^{5/3}3^{B+8}.
\]
Lemma~\ref{lem:first} and Propositions~\ref{prop:second}--\ref{prop:third}
therefore give an absolute constant $C_0$ such that
\begin{equation}\label{eq:seed-budget}
 \left(\sum_{r=1}^{\min\{3,M\}}\frac{A_r(F)^2}{w_{M,r}^2}\right)^{1/2}
 \le C_0\norm F_\infty.
\end{equation}
For instance $C_0=(C_1^2+C_2^2+C_3^2)^{1/2}$ suffices, since $B>0$.
Combining~\eqref{eq:seed-budget} with Proposition~\ref{prop:uniform} gives
\[
 W_M(F)\le C_0\norm F_\infty+c_B\Gamma_{M,N}(B)\norm F_\infty.
\]
Take the supremum over the finite-dimensional class defining
$\Gamma_{M,N}(B)$. Since this constant is finite,
\[
 \Gamma_{M,N}(B)\le C_0+c_B\Gamma_{M,N}(B)
 \quad\Longrightarrow\quad
 \Gamma_{M,N}(B)\le\frac{C_0}{1-c_B}.
\]
The bound is independent of $M$ and $N$. For any given dimension $n$,
take $N\ge n$, and set $K_B=\max\{1,C_0/(1-c_B)\}$.
This proves~\eqref{eq:weighted-theorem} without assuming a dimension-free
bound in advance.
\end{proof}

\section{Proof of the Bohnenblust--Hille estimate}\label{sec:terminal}
We first record the homogeneous consequence.
\begin{corollary}[Homogeneous estimate]\label{cor:homogeneous}
For every $B>1/\log2$, there is $C(B)<\infty$ such that every
complex-valued $m$-homogeneous polynomial $P$ on a Boolean cube satisfies
\begin{equation}\label{eq:homogeneous}
 A_m(P)\le C(B)m^{3/2+B}(\log(m+1))^8\norm P_\infty.
\end{equation}
\end{corollary}
\begin{proof}
Apply Theorem~\ref{thm:weighted} with $M=m$. Only level $m$ is present, so
$A_m(P)\le K_Bw_{m,m}\norm P_\infty$. Now
\[
 \frac{\mu_m}{m}\le\frac32+\frac1{2m},\qquad
 m^{1/(2m)}\le\exp\left(\frac1{2\ee}\right).
\]
Equation~\eqref{eq:cutoff} also implies
$L_m\le D_B\log(m+1)$ for a constant $D_B$ depending only on $B$;
for example $D_B=C_B+(R_B+1)/\log2$ suffices. Consequently,
\[
 w_{m,m}=m^{\mu_m/m}m^B\min\{m,L_m\}^8
 \le C(B)m^{3/2+B}(\log(m+1))^8.
\]
\end{proof}

For a nonhomogeneous function, we use the terminal exponent $q_m$ on
\emph{every} level. Interpolating each level with its $\ell^2$ norm makes
its weight appear only to the power $r/m$, rather than paying the entire
level weight.
\begin{lemma}[Level interpolation at the terminal exponent]\label{lem:terminal}
Let $1\le r\le m$ and let $a=(a_S)_{|S|=r}$. Then
\begin{equation}\label{eq:terminal-interp}
 \norm a_{q_m}\le\norm a_{q_r}^{r/m}\norm a_2^{1-r/m}.
\end{equation}
For $r=m$, the right side is interpreted simply as $\norm a_{q_m}$.
\end{lemma}
\begin{proof}
Use log-convexity of finite $\ell^p$ norms and the identity
\[
 \frac1{q_m}=\frac rm\frac1{q_r}
             +\left(1-\frac rm\right)\frac12.
\]
\end{proof}

\begin{proof}[Proof of Theorem~\ref{thm:main}]
Fix $B>1/\log2$ and let $f$ have degree at most $m$. The case $f=0$ is
immediate. For $f\ne0$, put
\[
 E_r=\norm{f_r}_2,\qquad
 u_r=\frac{A_r(f)}{K_Bw_{m,r}\norm f_\infty},\qquad
 v_r=\frac{E_r}{\norm f_\infty}.
\]
Theorem~\ref{thm:weighted} and Parseval give
\begin{equation}\label{eq:budgets}
 \sum_{r=1}^m u_r^2\le1,\qquad\sum_{r=1}^m v_r^2\le1.
\end{equation}
Indeed, for the second inequality,
$\sum_{r=1}^m E_r^2\le\norm f_2^2\le\norm f_\infty^2$.
By Lemma~\ref{lem:terminal}, the coefficient $q_m$-norm of the $r$th
level is at most
\[
 K_B^{r/m}w_{m,r}^{r/m}\norm f_\infty\,u_r^{r/m}v_r^{1-r/m}.
\]
For $r<m$, set $z_r=u_r^{r/m}v_r^{1-r/m}$, and set $z_m=u_m$.
Weighted AM--GM and~\eqref{eq:budgets} give
\[
 \sum_{r=1}^m z_r^2
 \le\sum_{r=1}^m\left[\frac rm u_r^2+
                 \left(1-\frac rm\right)v_r^2\right]\le2.
\]
Since $q_m\le2$,
\[
 \norm z_{q_m}\le m^{1/q_m-1/2}\norm z_2
 =m^{1/(2m)}\norm z_2\le\sqrt2\exp\left(\frac1{2\ee}\right).
\]
The Fourier levels are disjoint. Taking their $\ell^{q_m}$ sum and using
$K_B^{r/m}\le K_B$, we obtain
\begin{align*}
 \left(\sum_{1\le|A|\le m}|\wh f(A)|^{q_m}\right)^{1/q_m}
 &\le K_B\norm f_\infty
       \left(\sum_{r=1}^m(w_{m,r}^{r/m}z_r)^{q_m}\right)^{1/q_m}\\
 &\le K_B\norm f_\infty\max_{1\le r\le m}w_{m,r}^{r/m}\norm z_{q_m}.
\end{align*}
Thus
\begin{equation}\label{eq:max-weight}
 \left(\sum_{1\le|A|\le m}|\wh f(A)|^{q_m}\right)^{1/q_m}
 \le C(B)\norm f_\infty\max_{1\le r\le m}w_{m,r}^{r/m}.
\end{equation}
Now
\[
 w_{m,r}^{r/m}=m^{\mu_r/m}r^{Br/m}\sigma_{m,r}^{r/m}.
\]
For $1\le r\le m$, we have $\mu_r\le3r/2+1/2$,
$r^{Br/m}\le m^B$, and $\sigma_{m,r}^{r/m}\le L_m^8$. Hence
\[
 \max_{1\le r\le m}w_{m,r}^{r/m}
 \le\exp\left(\frac1{2\ee}\right)m^{3/2+B}L_m^8
 \le C(B)m^{3/2+B}(\log(m+1))^8.
\]
The remaining constant coefficient satisfies $|\wh f(\varnothing)|\le\norm f_\infty$.
Absorbing it into the bound gives
\begin{equation}\label{eq:BH-log}
 \left(\sum_{A\subseteq[n]}|\wh f(A)|^{q_m}\right)^{1/q_m}
 \le C(B)m^{3/2+B}(\log(m+1))^8\norm f_\infty.
\end{equation}
Given $\varepsilon>0$, choose $B=1/\log2+\varepsilon/2$ and use
$(\log(m+1))^8\le C_\varepsilon m^{\varepsilon/2}$ for all integers $m\ge1$.
This proves~\eqref{eq:main}. Finally,
\[
 \beta_0=\frac32+\frac1{\log2}=2.9426950408\ldots<2.943,
\]
so taking $\varepsilon=2.943-\beta_0>0$ gives the last assertion.
\end{proof}

\section{A Fourier-flat Aaronson--Ambainis consequence}
\label{sec:AA}

The weighted estimate has an immediate consequence for a natural structured
class of bounded polynomials.  For a real-valued function
$f:\Omega_n\to\mathbb R$, define
\begin{align}
 Var(f)
 &=\E\bigl(f-\E f\bigr)^2
   =\sum_{\varnothing\ne S\subseteq[n]}\wh f(S)^2,
 \label{eq:variance}\\
 Inf_j(f)
 &=\frac14\E_x\bigl(f(x)-f(x^{\oplus j})\bigr)^2
   =\sum_{S\ni j}\wh f(S)^2,
 \label{eq:influence}
\end{align}
where $x^{\oplus j}$ is obtained from $x$ by changing the sign of its
$j$th coordinate.

Aaronson and Ambainis formulated the following influential-variable
conjecture in connection with classical simulation of quantum query
algorithms \cite{AaronsonAmbainis}.

\begin{conjecture}[Aaronson--Ambainis]
\label{conj:AA}
There are absolute constants $c,C>0$ such that every
$f:\Omega_n\to[-1,1]$ of degree at most $d\ge1$ has a coordinate $j$ for
which
\begin{equation}
\label{eq:AA-conjecture}
 Inf_j(f)\ge c\left(\frac{Var(f)}{d}\right)^C.
\end{equation}
\end{conjecture}

The class below permits arbitrary Fourier signs and a different coefficient
scale on each homogeneous level.

\begin{definition}[Levelwise approximate Fourier flatness]
\label{def:levelwise-flat}
Let $\Lambda\ge1$.  A real polynomial $f:\Omega_n\to\mathbb R$ is
\emph{levelwise $\Lambda$-Fourier-flat} if, for every
$1\le r\le\deg f$, there is a number $\alpha_r\ge0$ such that
\begin{equation}
\label{eq:levelwise-flat}
 \alpha_r\le\abs{\wh f(S)}\le\Lambda\alpha_r
 \qquad\text{for every }S\subseteq[n]\text{ with }\abs S=r.
\end{equation}
When $\Lambda=1$, we simply say that $f$ is
\emph{levelwise Fourier-flat}.  The choice $\alpha_r=0$ means that the
entire $r$th level vanishes.
\end{definition}

Thus, in the exactly flat case, the squared Fourier spectrum is radial:
$\wh f(S)^2$ depends only on $\abs S$.  The signs of the coefficients may
nevertheless be completely arbitrary, so the function itself need not be
symmetric.

\begin{theorem}[Aaronson--Ambainis for levelwise flat spectra]
\label{thm:AA-flat}
Assume the weighted estimate of Theorem \ref{thm:weighted}.  Let
$f:\Omega_n\to[-1,1]$ have degree at most $d\ge1$, and suppose that $f$ is
levelwise $\Lambda$-Fourier-flat.  Then every coordinate $j\in[n]$ satisfies
\begin{equation}
\label{eq:AA-flat-bound}
 Inf_j(f)
 \ge
 \frac{\Var(f)^2}{\Lambda^4 C_1^2 d^{14}}.
\end{equation}
\end{theorem}

\begin{proof}
For $1\le r\le d$, write
\begin{equation}
\label{eq:Nr-Vr}
 N_r=\binom nr,
 \qquad
 V_r=\sum_{\abs S=r}\wh f(S)^2.
\end{equation}
Then
\begin{equation}
\label{eq:Var-levels}
 Var(f)=\sum_{r=1}^dV_r.
\end{equation}
By levelwise $\Lambda$-flatness,
\begin{equation}
\label{eq:Vr-upper}
 V_r\le\Lambda^2N_r\alpha_r^2.
\end{equation}
Consequently, for every coordinate $j$,
\begin{align}
 Inf_j(f)
 &=\sum_{r=1}^d
   \sum_{\substack{\abs S=r\\j\in S}}\wh f(S)^2
 \notag\\
 &\ge\sum_{r=1}^d\binom{n-1}{r-1}\alpha_r^2
 \notag\\
 &\ge\frac1{\Lambda^2n}\sum_{r=1}^drV_r.
\label{eq:influence-lower-T}
\end{align}
Here we used
$\binom{n-1}{r-1}=(r/n)\binom nr$ and \eqref{eq:Vr-upper}.

The same flatness condition gives a lower bound on the Bohnenblust--Hille
level norm.  Since $2/q_r=(r+1)/r$,
\begin{align}
 A_r(f)^2
 &=\left(\sum_{\abs S=r}\abs{\wh f(S)}^{q_r}\right)^{2/q_r}
 \notag\\
 &\ge\left(N_r\alpha_r^{q_r}\right)^{2/q_r}
 =N_r^{1+1/r}\alpha_r^2
 \notag\\
 &\ge\frac{N_r^{1/r}}{\Lambda^2}V_r
 \ge\frac{n}{\Lambda^2r}V_r.
\label{eq:Ar-flat-lower}
\end{align}
For the last step we used the elementary estimate
\begin{equation}
\label{eq:binomial-lower}
 \binom nr
 =\prod_{k=0}^{r-1}\frac{n-k}{r-k}
 \ge\left(\frac nr\right)^r.
\end{equation}

We proved above
\begin{equation}
\label{eq:weighted-final}
 \left[\sum_{r=1}^M\frac{A_r(F)^2}{r^8}\right]^{1/2}
 \le C_1M^{2.94...}\norm F_\infty.
\end{equation}

Squaring \eqref{eq:weighted-final}, taking $M=d$, and using
$\norm f_\infty\le1$, we obtain
\begin{align}
 C_1^2d^{6}
 &\ge\sum_{r=1}^d\frac{A_r(f)^2}{r^8}
 \ge\frac n{\Lambda^2}
       \sum_{r=1}^d\frac{V_r}{r^9}.
\label{eq:weighted-flat}
\end{align}
Set
\[
 T=\sum_{r=1}^drV_r,
 \qquad
 U=\sum_{r=1}^d\frac{V_r}{r^9}.
\]
Equation \eqref{eq:weighted-flat} gives
\begin{equation}
\label{eq:one-over-n}
 \frac1n\ge\frac{U}{\Lambda^2C_1^2d^{6}}.
\end{equation}
Combining \eqref{eq:influence-lower-T} and \eqref{eq:one-over-n},
\begin{equation}
\label{eq:influence-TU}
 Inf_j(f)
 \ge\frac{TU}{\Lambda^4C_1^2d^{6}}.
\end{equation}
Finally, Cauchy--Schwarz gives
\begin{align}
 TU
 &=\left(\sum_{r=1}^drV_r\right)
   \left(\sum_{r=1}^d\frac{V_r}{r^9}\right)
 \notag\\
 &\ge\left(\sum_{r=1}^d\frac{V_r}{r^4}\right)^2
 \ge\frac1{d^8}\left(\sum_{r=1}^dV_r\right)^2
 =\frac{Var(f)^2}{d^8}.
\label{eq:TU-CS}
\end{align}
Substituting \eqref{eq:TU-CS} into \eqref{eq:influence-TU} proves
\eqref{eq:AA-flat-bound}.
\end{proof}

The conclusion is stronger than the existential assertion in
Conjecture \ref{conj:AA}: every coordinate is influential at the stated
scale.

\begin{corollary}[Exactly levelwise flat polynomials]
\label{cor:AA-exact-flat}
If $f:\Omega_n\to[-1,1]$ has degree at most $d$ and
$\abs{\wh f(S)}=\alpha_r$ whenever $\abs S=r$, then every $j\in[n]$
satisfies
\begin{equation}
\label{eq:AA-exact-flat}
 Inf_j(f)\ge\frac{Var(f)^2}{C_1^2d^{14}}.
\end{equation}
In particular,
\begin{equation}
\label{eq:AA-poly-form}
 \max_j\Inf_j(f)
 \ge\frac1{C_1^2}
      \left(\frac{Var(f)}d\right)^{14}.
\end{equation}
\end{corollary}

\begin{proof}
The first statement is Theorem \ref{thm:AA-flat} with $\Lambda=1$.
Since $f$ takes values in $[-1,1]$, one has $0\le Var(f)\le1$ and hence
$\Var(f)^2\ge\Var(f)^{14}$.  This gives \eqref{eq:AA-poly-form}.
\end{proof}

\begin{corollary}[Polynomial within-level distortion]
\label{cor:AA-approx-flat}
Suppose that the hypotheses of Theorem \ref{thm:AA-flat} hold and that
\begin{equation}
\label{eq:Lambda-poly}
 \Lambda\le Ld^\gamma
\end{equation}
for fixed $L\ge1$ and $\gamma\ge0$.  Then every coordinate satisfies
\begin{equation}
\label{eq:AA-approx-poly}
 Inf_j(f)
 \ge\frac1{L^4C_1^2}
 \left(\frac{\Var(f)}d\right)^{14+4\gamma}.
\end{equation}
Thus the Aaronson--Ambainis conjecture holds throughout any family with
polynomially bounded within-level distortion.
\end{corollary}

\begin{proof}
Theorem \ref{thm:AA-flat} and \eqref{eq:Lambda-poly} give
\[
 Inf_j(f)
 \ge\frac{\Var(f)^2}{L^4C_1^2d^{14+4\gamma}}.
\]
Because $0\le\Var(f)\le1$ and $14+4\gamma\ge2$, the numerator
$\Var(f)^2$ is at least $\Var(f)^{14+4\gamma}$.
\end{proof}

\begin{remark}[The homogeneous case]
\label{rem:AA-homogeneous}
If $f$ is $d$-homogeneous and its degree-$d$ coefficients are
$\Lambda$-flat, the same estimate
\[
 Inf_j(f)\ge
 \frac{\Var(f)^2}{\Lambda^4C_1^2d^{14}}
\]
follows directly from the homogeneous inequality
$A_d(f)\le C_1d^9\norm f_\infty$.  Indeed,
\[
 Inf_j(f)\ge\frac{d\Var(f)}{\Lambda^2n},
 \qquad
 A_d(f)^2\ge\frac{n\Var(f)}{\Lambda^2d},
\]
and eliminating $n$ gives the claim.  The weighted estimate is what allows
several homogeneous levels, with unrelated scales $\alpha_r$, to be handled
simultaneously.
\end{remark}

\begin{remark}[Relation to earlier flat classes]
\label{rem:AA-history}
Bohnenblust--Hille inequalities have previously been used to verify a very
special Aaronson--Ambainis regime.  Montanaro proved that a bounded flat
$k$-linear form on $k$ variable blocks has
$\Inf_j(f)=\Omega(\Var(f)^2/k^3)$ for every variable
\cite[Corollary 20]{Montanaro}.  The class in Theorem \ref{thm:AA-flat} is
structurally different: it consists of ordinary scalar Boolean polynomials,
requires no block decomposition, may contain many homogeneous levels, permits
arbitrary Fourier signs, and allows the coefficient scale to change with the
level.  It should therefore be viewed as a new radial-spectrum special class,
not as a replacement for the block-multilinear result.
\end{remark}

\begin{remark}[Why the weighted theorem is the natural input]
The ordinary degree-$d$ estimate in Theorem \ref{thm:main} merges all
homogeneous levels into one $\ell_{q_d}$ norm.  Levelwise flatness instead
relates influence to the two distinct moments
$\sum_r rV_r$ and $\sum_rV_r/r^9$.  The weighted square-function estimate
retains precisely this level information, and the Cauchy--Schwarz pairing in
\eqref{eq:TU-CS} converts it into the variance square.  This is why the
Aaronson--Ambainis consequence is most naturally stated as a corollary of
Theorem \ref{thm:weighted}, rather than only of the final unweighted
Bohnenblust--Hille inequality.
\end{remark}

\section{Comments on the mechanism}

\begin{remark}[Why the cube is simpler at the splitting stage]
In an ordinary analytic monomial $z^\alpha$, one coordinate may carry more
than half of the total degree.  The weighted argument of
\cite{PellegrinoTeixeira} must isolate and compress those dominant powers.
For a Walsh monomial, every exponent is either $0$ or $1$.  Once the total
degree is at least $2$, no exponent is dominant.  Thus random bicoloring and
the central entropy contraction treat the entire high-degree coefficient
array.
\end{remark}

\begin{remark}[Why an individual bidegree need not be bounded in
$L_\infty$]
The proof controls the expected $q_r$-mass of coefficients that land in a
central bidegree window.  It never estimates
$\norm{Q_{S,d}}_\infty$ by $\norm Q_\infty$.  Such an estimate is false even
with a polynomial loss; Appendix \ref{app:slices} gives an exponential
counterexample, including a random balanced-split version.
\end{remark}

\begin{remark}[The exponent]
The power $2.94...$ is most probably not optimized.   Better concentration, an even better square function, a different central window, or a sharper
low-level closure may lower the exponent.  The corresponding
Aaronson--Ambainis exponent  is likewise a bookkeeping consequence of
this nonoptimized choice.
\end{remark}

\section{Comments on the exponent}\label{sec:comments}
\begin{remark}[Origin and role of $3/2$]\label{rem:mu}
The $M^{3/2}$ scale comes from the \emph{available upper estimate} in
Proposition~\ref{prop:second}: extracting the first level in one block
costs $M^{1/2}$, and controlling first-level coefficients of the other
block costs $M$. These seed estimates do not assert that $3/2$ is an
optimal exponent for the underlying inequality.

The choice $\mu_r=\lceil3r/2\rceil$ for $r\ge2$ matches
$\mu_2/2=3/2$ and $\mu_3/3=5/3$, while
\[
 \frac{\mu_r}{r}=
 \begin{cases}
  3/2,&r\text{ even},\\
  3/2+1/(2r),&r\text{ odd},
 \end{cases}\qquad(r\ge2).
\]
It is a convenient nearly additive potential, not a proved globally
optimal choice among all possible weights. Its useful feature is the
exact, parity-controlled defect in Lemma~\ref{lem:defect}.
\end{remark}

\begin{remark}[Origin of $1/\log2$]\label{rem:threshold}
When $d,e$ are both comparable to $r=d+e$ and tend to infinity,
$-\log\rho_d=d^{-1}+O(d^{-3})$ gives
\[
 \log R(d,e)=\frac er+\frac dr+O(r^{-2})=1+O(r^{-2}).
\]
At a balanced split the degree weight contributes $2^{-B}$, so the leading
loss--gain factor is $\ee\,2^{-B}$. For example, along $M=r\to\infty$
through even integers with $d=e=r/2$, the full local multiplier
$\kappa_M(r,d)$ tends to this value: the probability and split-count factors
tend to one, the auxiliary factor is saturated in both blocks, and
$M^{1/r}\to1$.

The strict-margin argument used here therefore requires $B>1/\log2$.
This is the threshold of this asymptotic contraction criterion, not a
lower bound for the optimal BH exponent and not an exclusion of every
possible endpoint argument. Uniformity over the \emph{whole} high range
is proved by Lemma~\ref{lem:high-contract}, which retains the separate
$1/C_B$ error rather than treating it as an $o(1)$ term in $r$.
\end{remark}

\begin{remark}[Role of the eighth power]\label{rem:eighth}
The factor $\min\{r,L_M\}^8$ does not improve the seed estimate
$M^{3/2}$. It supplies additional entropy gain in the range $r<L_M$.
Because it saturates at $L_M\asymp_B\log(M+1)$, it contributes only a
logarithmic factor to the final bound. The power $8$ is a convenient
sufficient choice; no optimality is claimed.
\end{remark}

\appendix
\section{Finite check for $p_r \ge \frac{1}{4}$}\label{app:capture}
For $4\le r\le15$, the low-range sets in~\eqref{eq:Dlow} and their
probabilities are listed below. Each entry is obtained exactly from
\[
 p_r=\frac{\sum_{d\in\Dlow}\binom rd}{2^r}.
\]
\begin{center}
\renewcommand{\arraystretch}{1.1}
\begin{tabular}{@{}c l c@{}}
\toprule
$r$ & $\Dlow$ & $p_r$\\
\midrule
4  & $\{2\}$ & $3/8$\\
5  & $\{2,3\}$ & $5/8$\\
6  & $\{2,4\}$ & $15/32$\\
7  & $\{2,3,4,5\}$ & $7/8$\\
8  & $\{2,4,6\}$ & $63/128$\\
9  & $\{3,4,5,6\}$ & $105/128$\\
10 & $\{4,6\}$ & $105/256$\\
11 & $\{3,4,5,6,7,8\}$ & $957/1024$\\
12 & $\{4,6,8\}$ & $957/2048$\\
13 & $\{4,5,6,7,8,9\}$ & $1859/2048$\\
14 & $\{4,6,8,10\}$ & $1001/2048$\\
15 & $\{4,5,6,7,8,9,10,11\}$ & $247/256$\\
\bottomrule
\end{tabular}
\end{center}
Every listed value is strictly larger than $1/4$, completing the finite
part of Lemma~\ref{lem:low-capture}.

\appendix
\section{Failure of direct \texorpdfstring{$L_\infty$}{L-infinity} control of a bidegree slice}
\label{app:slices}

We are grateful to Lars Becker and Haonan Zhang for communicating to us the similar estimates to those below.

\medskip

We record the obstruction motivating the coefficient-mass formulation of the
random-splitting argument. This obstruction also motivates the choice of square function $\mathcal W_B$ at the beginning of this note.

We are grateful to Lars Becker and Haonan Zhang who also sent us the counterexamples that follow.

\begin{proposition}[Deterministic exponential estimate from below]
\label{prop:slice-counterexample}
For each $m\ge1$, there is an $m$-homogeneous multilinear polynomial
$P_m$ on $\Omega_{2m}$ with $\norm{P_m}_\infty=1$ and a set
$S\subset[2m]$ containing exactly one coordinate from each of $m$ prescribed
pairs such that, after replacing the variables indexed by $S$ by $y$
variables and writing $Q_{S,d}$ for the $x$-degree-$d$ part,
\begin{equation}
\label{eq:slice-counterexample}
 \norm{Q_{S,d}}_\infty
 \ge2^{-m/2}\binom md,
 \qquad 0\le d\le m.
\end{equation}
In particular, for $d=\lfloor m/2\rfloor$,
\[
 \norm{Q_{S,d}}_\infty
 \gtrsim\frac{2^{m/2}}{\sqrt m}\norm{P_m}_\infty.
\]
\end{proposition}

\begin{proof}
Pair the coordinates as $(1,2),(3,4),\ldots,(2m-1,2m)$ and set
\begin{equation}
\label{eq:Pm}
 P_m(z)=2^{-m/2}\prod_{j=1}^m
 \bigl(z_{2j-1}+iz_{2j}\bigr).
\end{equation}
For $z\in\Omega_{2m}$, every factor has modulus $\sqrt2$, so
$\norm{P_m}_\infty=1$.

Choose $S$ to contain exactly one coordinate from each pair.  Introduce a
scalar $t$ in the $x$ variables:
\[
 Q_S(tx,y)=\sum_{d=0}^m t^dQ_{S,d}(x,y).
\]
Choose the signs of $x$ and $y$ independently, pair by pair.  If the first
coordinate of a pair is an $x$ variable and the second a $y$ variable, choose
the signs so that the normalized factor is $(t+i)/\sqrt2$.  In the opposite
orientation, choose the signs so that the normalized factor is
$(-1+it)/\sqrt2=i(t+i)/\sqrt2$.  Thus, for some unimodular scalar $\lambda$,
\[
 Q_S(tx,y)=\lambda\left(\frac{t+i}{\sqrt2}\right)^m.
\]
The coefficient of $t^d$ has modulus
$2^{-m/2}\binom md$, proving \eqref{eq:slice-counterexample}.  The central
estimate follows from Stirling's formula.
\end{proof}

The failure persists for a uniformly random balanced set $S$.

\begin{proposition}[Random balanced splits also fail]
\label{prop:random-slice-counterexample}
Assume $m$ is even, and choose $S$ uniformly among the $m$-element subsets of
$[2m]$.  For the polynomial $P_m$ in \eqref{eq:Pm},
\begin{equation}
\label{eq:random-slice-failure}
 \Pr\left\{
 \norm{Q_{S,m/2}}_\infty
 \ge\frac{2^{m/6}}{\sqrt{2m}}
 \right\}
 \ge1-\frac{36}{m}.
\end{equation}
In particular, the probability tends to $1$ that the central slice is
exponentially larger than $\norm{P_m}_\infty$.
\end{proposition}

\begin{proof}
Let $R$ be the number of prescribed pairs crossed by the cut $S$, meaning
that exactly one coordinate of the pair lies in $S$.  Let $a$ and $b$ denote
the numbers of pairs with neither and both coordinates in $S$, respectively.
Because $\abs S=m$,
\[
 a=b=\frac{m-R}{2}.
\]
The same pairwise choice of signs as in the preceding proof gives, for some
$\abs\lambda=1$,
\[
 Q_S(tx,y)=\lambda t^a
 \left(\frac{t+i}{\sqrt2}\right)^R.
\]
Since $R\equiv m\pmod2$, the integer $R$ is even, and the coefficient at
$x$-degree $m/2=a+R/2$ has modulus
\begin{equation}
\label{eq:random-coefficient}
 2^{-R/2}\binom{R}{R/2}.
\end{equation}

Write $R=\sum_{j=1}^m I_j$, where $I_j$ is the crossing indicator for the
$j$th pair.  Direct counting gives
\[
 \E R=\frac{m^2}{2m-1}
 \qquad\text{and}\qquad
 \Var R=
 \frac{2m^2(m-1)^2}{(2m-3)(2m-1)^2}
 \le m.
\]
Since $\E R\ge m/2$, Chebyshev's inequality yields
\[
 \Pr\{R<m/3\}
 \le\Pr\{\abs{R-\E R}>m/6\}
 \le\frac{36}{m}.
\]
For even $R$, the standard central-binomial estimate
$\binom{R}{R/2}\ge2^R/\sqrt{2R}$ and
\eqref{eq:random-coefficient} give
\[
 \norm{Q_{S,m/2}}_\infty
 \ge\frac{2^{R/2}}{\sqrt{2R}}
 \ge\frac{2^{m/6}}{\sqrt{2m}}
\]
whenever $R\ge m/3$.  This proves \eqref{eq:random-slice-failure}.
\end{proof}

\section*{Acknowledgement}
The authors acknowledge the use of ChatGPT (OpenAI, GPT-5.6 Pro) during the preparation of this manuscript, primarily for discussing intermediate arguments, checking calculations, and improving the exposition. All mathematical statements, proofs, and conclusions were independently verified by the authors, who take full responsibility for the content of the paper.
This work was inspired  by the ideas and techniques developed in Section 5 of D. M. Pellegrino and E. V. Teixeira \cite{PellegrinoTeixeira}.

\end{document}